\documentclass[11pt, a4paper, oneside,reqno]{amsart}

\usepackage{style}
\usepackage[square,numbers]{natbib}
\usepackage{amsmath}
\usepackage{amssymb}
\usepackage{mathtools}
\usepackage{bm}
\usepackage{comment}
\usepackage{booktabs}

\usepackage{tikz}
\usepackage{pgfplots}
\usetikzlibrary{calc}
\pgfplotsset{compat=1.18}

\usepackage[textsize=tiny]{todonotes}

\usepackage[linesnumbered,ruled]{algorithm2e}
\usepackage{graphicx}
\usepackage{subcaption}

\usepackage{amsthm}
\newtheorem{theorem}{Theorem}[section]
\newtheorem{defn}[theorem]{Definition}
\newtheorem{assum}[theorem]{Assumption}
\newtheorem{prop}[theorem]{Proposition}
\newtheorem{cor}[theorem]{Corollary}
\newtheorem{lem}[theorem]{Lemma}
\newtheorem{rem}[theorem]{Remark}
\newtheorem{exmp}[theorem]{Example}

\newtheorem*{claim*}{Claim}

\newtheorem*{pf}{Proof}

\graphicspath{{Fig/}}

\allowdisplaybreaks

\title[]{Learning in Inverse Games: \\ Tractable Training with Probabilistic Guarantees}

\author[]{Arghya Mallick, Reza Rahimi Baghbadorani, Peyman Mohajerin Esfahani, and Sergio Grammatico }
\thanks{An article based on Section \ref{dg_game_section} of this work is accepted for oral presentation at the $65$th IEEE Conference on Decision and Control (CDC), 2026, in Hawaii, USA. This work was partly funded by the ERC project ARGON (grant no. 101231127).}
\thanks{\textcolor{black}{Arghya Mallick and Sergio Grammatico are with Delft University of Technology; e-mails: \href{mailto:A.Mallick@tudelft.nl}{\texttt{\{A.Mallick,s.grammatico\}@tudelft.nl}}. Reza Rahimi Baghbadorani is with Erasmus University; e-mail: \href{mailto:rahimibaghbadorani@rsm.nl}{\texttt{rahimibaghbadorani@rsm.nl}}. Peyman Mohajerin Esfahani is with University of Toronto; e-mail: \href{mailto:P.MohajerinEsfahani@tudelft.nl}{\texttt{P.MohajerinEsfahani@utoronto.ca}}.}}

\begin{document}

\begin{abstract}
Inverse game theory seeks to learn agents' unknown objectives from observed equilibrium behavior. Existing residual-based approaches can lead to non-convex problems and need not ensure strong monotonicity of the learned game, limiting reliable equilibrium prediction. We develop a tractable convex framework for learning static and dynamic non-cooperative games using first-order and Nikaido-Isoda (NI) loss. Specialized operator parameterizations, including a Helmholtz--Hodge decomposition in a reproducing kernel Hilbert space (RKHS), enforce strong monotonicity for quadratic, nonparametric, non-quadratic, and linear-quadratic dynamic games. We establish finite-sample out-of-sample prediction guarantees, using Rademacher complexity to match existing bounds for the non-quadratic setting. For noisy sequential data, we develop a robust receding-horizon learning scheme whose adaptive regularization admits a maximum a posteriori (MAP) interpretation and is computed using randomized trace estimation. Numerical experiments, including a stylized autonomous-vehicle collision-avoidance application, demonstrate predictive accuracy and robustness to measurement noise.
\end{abstract}

\maketitle

\section{Introduction}
\label{sec:intro} In multi-agent non-cooperative environments, self-interested entities optimize individual utilities that depend on the joint actions of their peers. This strategic interaction is naturally formalized as a non-cooperative game. Across diverse applications, such as network resource allocation \cite{johari2005efficiency}, peer-to-peer energy trading in power systems~\cite{belgioioso2022operationally}, and autonomous driving \cite{peters2023online}, these interactions are modeled as games wherein the emergent operational behavior coincides with a Nash equilibrium. From the perspective of a central system operator, anticipating these equilibrium profiles is critical for ensuring operational efficiency and system reliability. This raises a fundamental system-theoretic question: can a designer reliably infer the underlying structure of a game from observations of its equilibria? More specifically, can the private objective functions of the agents be learned? While inverse optimization theory \cite{mohajerin2018data,zattoni2025learning} affirmatively resolves this for single-agent systems, \textit{inverse game theory} extends this paradigm to multi-agent settings. It establishes a rigorous mathematical framework to learn the unknown objective functions of self-interested agents using historical datasets of exogenous signals and their corresponding Nash equilibrium actions.

To illustrate the practical utility of inverse game theory, consider network resource allocation \cite{johari2005efficiency} within cloud computing infrastructures. The system operator is the cloud provider, while the competing agents comprise bandwidth-intensive services, such as video transcoders and database synchronizers. These agents seek to maximize their individual data throughput across a shared, capacity-constrained network, thereby inducing a non-cooperative resource-sharing game. By continuously monitoring exogenous network signals (e.g., background traffic loads) alongside the resulting equilibrium bandwidth allocations, the provider can utilize inverse game theory to learn the hidden utility functions of these services. Consequently, when confronted with a disruptive future state, such as an anticipated surge in prime-time traffic, the provider can accurately forecast the emergent Nash equilibrium actions. This predictive capability enables the proactive implementation of traffic-shaping measures to mitigate congestion and avert probable network outages. 

The problem of inferring objective functions from observed behavior traces its roots to single-agent inverse optimal control and inverse reinforcement learning (IRL) \cite{kalman1964linear,ng2000algorithms}. In continuous control domains, bilevel optimization frameworks are widely adopted. These approaches typically embed the Karush-Kuhn-Tucker (KKT) optimality conditions of the forward problem as constraints \cite{mombaur2010human} or directly minimize the condition residuals evaluated at the demonstrated trajectories \cite{englert2017inverse}. While extensions to multi-agent IRL exist, they largely assume potential game structures and do not generalize to arbitrary non-cooperative settings. To address this, recent literature has formalized inverse dynamic games \cite{awasthi2020inverse,li2025inverse,yu2022inverse}. A prevalent methodology extends single-agent residual minimization by penalizing the violation of first-order necessary conditions for open-loop Nash equilibria \cite{awasthi2020inverse}. Beyond resulting in highly non-convex optimization problems~\cite{li2025inverse}, an important limitation of these methods is the absence of inherent guarantees that the learned parameters will yield a monotone game during forward prediction. Moreover, the methodology in~\cite{yu2022inverse} is limited to an unconstrained two-player dynamic game. In a foundational contribution aligned with inverse game theory, the authors of \cite{bertsimas2015data} proposed a tractable convex formulation for inverse variational inequalities, offering both parametric and nonparametric approaches to estimate the pseudogradient operator. However, their methodology does not enforce monotonicity on the learned pseudogradient operator, a critical property required for the convergence of standard Nash equilibrium-seeking algorithms \cite{salehisadaghiani2016distributed}. Overcoming limitations of existing approaches \cite{awasthi2020inverse,bertsimas2015data,li2025inverse,yu2022inverse}, we propose a tractable inverse game formulation equipped with formal statistical guarantees for predicting the Nash equilibria of the learned game. Specifically, our main contributions are as follows. 
\begin{enumerate}[leftmargin=1.5em]
\renewcommand{\labelenumi}{(\roman{enumi})}
    \item \textbf{Tractable convex formulations:}  We develop tractable convex formulations for three classes of inverse games: 
    \begin{enumerate}[leftmargin=0.8em]
        \item \textit{Exact reformulation with NI loss.} For inverse linear-quadratic games, we introduce the \textit{Nikaido--Isoda (NI) loss} and derive an exact convex reformulation with LMI constraints (Proposition~\ref{prop_qg}). 
        
        \item \textit{Non-quadratic monotone games.} For inverse non-quadratic monotone games, we parameterize the game operator in an RKHS using the Helmholtz--Hodge decomposition and reduce the resulting infinite-dimensional inverse problem to a finite-dimensional convex program with second-order cone constraints (Theorem~\ref{kernel_formulation}). 
        
        \item \textit{Noisy data in dynamic games.} For inverse linear-quadratic dynamic games, we enforce strong monotonicity via LMIs and mitigate noise propagation in sequential data through a Tikhonov-regularized receding-horizon formulation. With an appropriate dynamic update of the regularization penalty, its solution coincides with the Bayesian MAP estimate of the noise-free parameter (Proposition \ref{prop:MAP-equivalence}). 
    \end{enumerate}
    \item \textbf{Probabilistic guarantees:} For both inverse quadratic and non-quadratic monotone games, we establish finite-sample guarantees on the out-of-sample predictive accuracy of the learned forward game (Theorems \ref{qg_theorem} and \ref{nqg_theorem}). For the latter, exploiting strong monotonicity and Rademacher complexity yields a statistical bound matching the best-known guarantee for the inverse variational inequality framework \cite[Theorem~8]{bertsimas2015data}.
\end{enumerate}
\subsection*{Notation}
The set $\mathbb{R}$ and $\mathbb{R}^n$ denote the reals and $n$-dimensional Euclidean space; $\mathbb{S}^n_{\succeq0}$ ($\mathbb{S}^n_{\succ0}$) the symmetric positive (semi-)definite matrices. We use $\|\cdot\|$, $\|\cdot\|_{\rm Fr}$, and $\|\cdot\|_{\rm op}$ for the Euclidean, Frobenius, and operator norms, respectively; $\langle\cdot,\cdot\rangle$ and $\mathrm{Tr}(\cdot)$ denote the inner product and trace. The operators $\mathrm{col}(\cdot)$ and $\mathrm{diag}(\cdot)$ denote column stacking and block diagonalization. For twice differentiable $f$, $\nabla_x f$ and $\nabla_x^2f$ denote its gradient and Hessian. Let $\mathcal I:=\{1,\ldots,N\}$, $u=\mathrm{col}((u^i)_{i\in\mathcal I})$, and $u^{-i}:=\mathrm{col}((u^j)_{j\in\mathcal I\setminus\{i\}})$. Finally, $\mathbb P[A]$ and $\mathbb E[\chi]$ denote probability and expectation, and $[T]=\{1,\ldots,T\}$ for $T\in\mathbb N$. The variational inequality problem, denoted by VI($F,\mathcal{U}$), is to find $u^* \in \mathcal{U}$ such that $F(u^*)^{\top}(u-u^*) \geq 0, \; \forall u \in \mathcal{U}$.
\section{Problem Description}
Let us first formalize the forward game setting. We denote the exogenous signal by $s \in \mathbb{S}$, where $\mathbb{S}$ is the signal space. Consider $N$ agents, indexed within the set $\mathcal{I}:= \{1,2, \dots, N\}$, engaged in the following non-cooperative forward game:
\begin{equation} \label{for_game}
    \forall i \in \mathcal{I}: \begin{cases}
         \min\limits_{u^i} &    J^i(u^i,u^{-i} \mid s), \\
        \:\text{s.t.}  &u^i \in \mathcal{U}^i(s),
    \end{cases}
\end{equation}
where $u^i \in \mathbb{R}^{n_i} $ is the decision vector of agent $i \in \mathcal{I}$, $u^{-i}= \mathrm{col}((u^j)_{j\in \mathcal{I}\setminus\{i\}})$ denotes the concatenated decision variables of all agents except $i$, $\mathcal{U}^i(s) \subseteq \mathbb{R}^{n_i}$ is the local feasible set of agent $i$ parameterized by signal $s$, and the cost function of agent $i$ is denoted by $J^i: \mathbb{S}\times \prod_{i \in \mathcal{I}}( \bigcup_{s \in \mathbb{S}}\mathcal{U}^i(s)) \rightarrow \mathbb{R}$, which depends on the decision variables of the other agents. Furthermore, the collective feasible set of the game, parameterized by the signal $s$, is denoted by $\mathcal{U}(s):= \prod_{i \in \mathcal{I}}\mathcal{U}^i(s)$. Given a signal $s \in \mathbb{S}$, the realized solution of \eqref{for_game} is assumed to be in a Nash equilibrium, defined as follows. 

\begin{defn}[Nash equilibrium] \label{NE defn}
    Given a signal $\tilde{s} \in \mathbb{S}$, a joint strategy profile $\tilde{u}:= \mathrm{col}((\tilde{u}^{i})_{i \in \mathcal{I}})$ is a Nash equilibrium of the game in \eqref{for_game} if $\tilde{u} \in \mathcal{U}(\tilde{s}) $, and for each $i \in \mathcal{I}$
    \begin{equation}
        J^i(\tilde{u}^i,\tilde{u}^{-i}\mid \tilde{s}) \leq J^i(u^i,\tilde{u}^{-i}\mid \tilde{s})
    \end{equation}
    for any $u^i \in \mathcal{U}^i(\tilde{s})$. 
\end{defn}
As a sequence of $T$ forward games in \eqref{for_game} unfolds, we let the observer collect a dataset of signal and Nash equilibrium pairs, denoted by $\mathcal{D}:= \{(\tilde{s}_t,\tilde{u}_t)\}_{t=1}^T$. Strictly speaking, the aim of the inverse game in our setting is to learn the cost functions $J^i$ for all $i \in \mathcal{I}$. Our solution methodologies follow a supervised learning approach instead of classic system identification methods, which implies that our methods do not intend to recover/identify the exact parameters of the objective function. Instead, we learn the forward game so that, given a new signal, it predicts the true Nash equilibrium with reasonable accuracy. To achieve this, we first search within a restricted space of functions to model $J^i$. Specifically, we aim to parameterize $J^i$ with a finite number of parameters and then develop a tractable formulation for learning them. To this end, we first impose the following standalone assumptions that hold throughout the article. 
\begin{assum}[Data consistency] \label{feas_ass} 
On the feasibility set of \eqref{for_game}, and dataset generation process, we assume the following:
    \begin{enumerate}\renewcommand{\labelenumi}{(\roman{enumi})}
        \item \emph{i.i.d samples.}  The training and testing signal data $(\tilde{s}_t,\tilde{u}_t)_{t\in \mathbb{N}}$ are i.i.d. realizations of a random variable $(\mathbf{s},\mathbf{u})$ drawn from an unknown probability measure $\mathbb{P}$.
        \item \emph{Geometry of feasibility set.} Given an arbitrary signal $s \in \mathbb{S}$, the local feasible set for the forward game \eqref{for_game} is characterized by bounded linear inequalities as $\mathcal{U}^i(s)= \{u^i \in \mathbb{R}^{n_i} \mid G_i(s)u^i \leq h_i(s)\}, \: \forall i\in \mathcal{I}$, where $G_i(s)\in \mathbb{R}^{m \times n_i}$ and $h_i(s)\in \mathbb{R}^{m}$.
        \item \emph{Observer's data.} The set $\mathcal{U}^i(s)$ is non-empty while the parameters $(G_i(s),h_i(s))_{i \in \mathcal{I}}$ are known to the observer.  
    \end{enumerate}
\end{assum}
The first assumption is critical for providing any statistical guarantees on its predictive performance. The last two assumptions are required for developing a tractable convex formulation for the training process. Although one could replace the polytopic feasibility set with any other convex feasibility set, we consider it for its relatively easier dual reformulation advantage.

We first study linear-quadratic forward games, deriving a convex training formulation with a strong performance guarantee. Section \ref{non_quad_sec} extends the framework to non-quadratic games with broader objective classes, enforcing strong monotonicity and providing an out-of-sample guarantee despite the infinite-dimensional function class. Section \ref{dg_game_section} develops a tractable, robust training formulation for linear-quadratic dynamic games with noisy data, while Section \ref{numerical_sec} numerically validates the main theoretical results. 
\section{Inverse Formulation for Linear Quadratic Games}
In the class of linear-quadratic games, the cost functions are defined as follows: 
\begin{equation} \label{obj_qg}
    J^i(u^i,u^{-i}\mid s):= \frac{1}{2} (u^i)^{\top} \left(R^i \sum_{j\in \mathcal{I} } u^j\right) + r^i(s)^{\top} u^i, \:\: \forall i \in \mathcal{I},
\end{equation}
where $r^i(s):= \Phi^i_r(s)r^i$, $R^i \in \mathbb{S}^{n_i}_{\succeq 0}$, and $r^i \in \mathbb{R}^{n_i}$. Here, $R^i$ and $r^i$ are the unknown parameters while $\Phi^i_r:\mathbb{S}\rightarrow \mathbb{R}^{n_i\times n_i}$ is a known feature matrix map. Without loss of generality, we consider $n_i=n_j$ for any $i,j\in \mathcal{I}$. To learn $J^i$, we parameterize it with $\theta$ as $J^i_{\theta}$ such that $(R^i_{\theta},r^i_{\theta})$ belong to the following hypothesis space:
\begin{align}\label{omega_qg}
    \Omega:=  \{R^i_{\theta} \in \mathbb{S}^{n_i}, r^i_{\theta}\in \mathbb{R}^{n_i} \mid \lambda_{l} I_{n_i} \preceq R^i_{\theta} \preceq \lambda_{u} I_{n_i}, \forall i \in \mathcal{I} \},
\end{align}
where the constants $0<\lambda_{l} < \lambda_{u}$ are chosen by user. Thus, we denote $r^i_{\theta}(s)= \Phi^i_r(s)r^i_\theta $. Furthermore, to restrict~\eqref{obj_qg} within the class of strongly monotone games, we need to ensure that the Jacobian of the pseudogradient operator is positive definite. Let $F_{\theta}(s,u):= \mathrm{col}(\nabla_{u^i} J^i_\theta(u^i,u^{-i} \mid s))_{i\in \mathcal{I}}$ be the pseudogradient operator. The Jacobian evaluates to $\nabla_u\big(F_\theta(s,u)\big)= \frac{1}{2}\big(\mathbf{R}(\theta)+ \mathbf{R}(\theta)^\top\big)$, where the matrix $\mathbf{R}(\theta) =\big[\tfrac{1+\delta_{ij}}{2}R_\theta^i\big]_{i,j=1}^N$ with $\delta_{ij}$ being the Kronecker delta, i.e., $\delta_{ij}=1$ if $i=j$, otherwise it is set as zero. We will use the Jacobian as a constraint later in this section. At the heart of any supervised learning methodology, the choice of a suitable loss function is pivotal. To optimize over $\theta$, we therefore introduce the following loss formulations.  

\subsection{Loss functions}
The \textit{Nikaido-Isoda function} \cite[Equation $(1)$]{nikaido1955note} plays a fundamental role in game theory, originally serving as a tool to establish existence results for the Nash equilibrium problems. We demonstrate how a meaningful loss function can be derived from it. First, we define
\begin{defn}[Nikaido-Isoda function \cite{nikaido1955note}]
    Given a signal $s \in \mathbb{S}$, the \textit{Nikaido-Isoda} (NI) function of the Nash equilibrium problem in \eqref{for_game} is
    \begin{equation} \label{NI_map}
        \Psi_{\theta}(x,y \mid s):= \sum_{i=1}^N \Big[J^i_{\theta}(x^i,x^{-i} \mid s)-J^i_{\theta}(y^i,x^{-i} \mid s)\Big],
    \end{equation}
    where $x=\mathrm{col}((x^i)_{i \in \mathcal{I}})$, and $y=\mathrm{col}((y^i)_{i \in \mathcal{I}})$ denotes joint strategies of the players. 
\end{defn}
The mapping in \eqref{NI_map} evaluates the cumulative improvement in the objective functions when each agent $i$ unilaterally changes their action from $x^i$ to $y^i$ while the other agents' actions remain fixed at $x^{-i}$. Furthermore, we introduce $L_{\theta}(s,x):= \sup_{y \in \mathcal{U}(s)} \: \Psi_{\theta}(x,y \mid s),$ and invoke~\cite[Theorem $3.2$]{facchinei2010generalized} to state the following properties: 
\begin{enumerate}\renewcommand{\labelenumi}{(\roman{enumi})}
    \item $L_{\theta}(s,x) \geq 0$ for all $x \in \mathcal{U}(s)$, and
    \item $L_{\theta}(s,\tilde{x}) = 0$ for $\tilde{x}\in \mathcal{U}(s)$ if and only if $\tilde{x}$ is a Nash equilibrium solution to the problem \eqref{for_game}.
\end{enumerate}
Leveraging the dataset of Nash equilibrium-signal pairs in $\mathcal{D}$, we evaluate $L_{\theta}$ at these points and propose the \textit{Nikaido-Isoda} (NI) loss function
\begin{align} \label{NI_loss}
    L^{\mathrm{NI}}_{\theta}(\tilde{s},\tilde{u}) :=\sum_{i=1}^{N} \left[ J^i_{\theta}(\tilde{u}^i,\tilde{u}^{-i} \mid \tilde{s})- \inf_{y^i \in \mathcal{U}^i(\tilde{s})} J^i_{\theta}(y^i,\tilde{u}^{-i} \mid \tilde{s}) \right],
\end{align}
where $\tilde{u}=(\tilde{u}^i)_{i \in \mathcal{I}}$. The suitability of \eqref{NI_loss} as a loss function stems directly from the implicit connection between the Nash equilibrium Definition~\ref{NE defn} and the function $L_{\theta}(s,x)$. Both encapsulate the impossibility of improving the objective value unilaterally. Since our goal is to learn the players' objective functions to accurately predict a Nash equilibrium given $s$, minimizing \eqref{NI_loss} naturally drives the parameters toward configurations that satisfy the equilibrium conditions. Interestingly, a parallel concept known as \emph{suboptimality loss} was proposed in \cite{mohajerin2018data} for inverse optimization problems, closely mirroring the proposed NI loss function in \eqref{NI_loss}. Alongside NI loss, authors in \cite{bertsimas2015data} explored a variational inequality-based approach, leading to the following alternative metric: 
\begin{defn}[First-order loss]\label{fo_loss_defn}
    Let $F_{\theta}(s,u):= \mathrm{col}(\nabla_{u^i} J^i_\theta(u^i,u^{-i} \mid s))_{i\in \mathcal{I}}$ be the pseudogradient operator of the forward game in \eqref{for_game}. The first-order (FO) loss is given by
    \begin{align} \label{fo_loss}
        L^{\mathrm{FO}}_{\theta}(\tilde{s},\tilde{u}):= \: \sup_{y \in \mathcal{U}(\tilde{s})} \: \langle F_{\theta}(\tilde{s},\tilde{u}), \tilde{u}-y \rangle.
    \end{align}
\end{defn}
Our framework favors optimizing the NI loss function over the FO loss due to the following dominance property. 
\begin{lem}[Dominance of NI loss]\label{loss_dominance}
    Let $L^{\mathrm{NI}}_{\theta}$ and $L^{\mathrm{FO}}_{\theta}$ denote the Nikaido-Isoda and first-order losses, respectively defined in \eqref{NI_loss} and \eqref{fo_loss}. If $F_{\theta}$ is a monotone pseudogradient operator, then
    \begin{equation}
        L^{\mathrm{NI}}_{\theta}(s,u) \leq L^{\mathrm{FO}}_{\theta}(s,u) \quad \forall s \in \mathbb{S}, \: u \in \mathcal{U}(s).
    \end{equation}
\end{lem}
The proof is provided in Appendix \ref{loss_lemma}. Having established the theoretical background for the loss formulation, we now present a tractable formulation of an inverse game problem. 

\subsection{Tractable formulation}
A standard methodology in supervised learning is to formulate an empirical risk minimization problem. Given the training dataset $\mathcal{D}$, the inverse game reduces to minimizing the empirical loss
\begin{subequations} \label{ig_main}
    \begin{align}
        \min_{\theta \in \Omega}\quad & \sum_{t=1}^{T} L^{\mathrm{NI}}_{\theta}(\tilde{s}_t,\tilde{u}_t), \\
    \mathrm{s.t.} \quad & \frac{1}{2}\big(\mathbf{R}(\theta)+ \mathbf{R}(\theta)^\top\big) \succeq \mu I \label{sm_cons}  
    \end{align}
\end{subequations}
where $\Omega$ is defined in~\eqref{omega_qg}, $\mu>0$ is the user-defined strong monotonicity parameter of the pseudogradient operator constraint in~\eqref{sm_cons}, and $L^{\mathrm{NI}}_\theta$ is defined in~\eqref{NI_loss}.  Due to the inner minimization embedded in the NI loss formulation, tractability of \eqref{ig_main} is not obvious despite the quadratic structure of $J^i_{\theta}$ in \eqref{obj_qg}. The following proposition establishes an exact convex reformulation of \eqref{ig_main}. 

\begin{prop}[Exact reformulation]\label{prop_qg}
    Consider the optimization problem \eqref{ig_main} where the candidate function $J^i_{\theta}$ admits the structure in \eqref{obj_qg}, and let Assumption \ref{feas_ass} hold. The empirical NI loss minimization problem in \eqref{ig_main} is equivalent to the following semidefinite program: 
    \begin{subequations} \label{main_qg_prob}
        \begin{align}
            \min_{\substack{\theta \in \Omega, \{\lambda_t, \\ \gamma_t\}_{t\in [T]}}} \quad & \sum_{t=1}^{T}\sum_{i=1}^{N} \Big[ \frac{1}{2}(\tilde{u}^i_t)^{\top} R^i_{\theta} \sum_{j=1}^N \tilde{u}^j_t + (r^i_{\theta}(\tilde{s}_t))^{\top} \tilde{u}^i_t + \nonumber \\ 
            & \gamma^i_t +\frac{1}{8}((\hat{u}^{-i}_t)^{\top} R^i_{\theta}\hat{u}^{-i}_t + 4(\hat{u}^{-i}_t)^{\top} r^i_{\theta}(\tilde{s}_t)) \Big] \\
            \mathrm{s.t.} \qquad & \begin{bmatrix}
                R^i_{\theta} & (r^i_{\theta}(\tilde{s}_t)+G_i(\tilde{s}_t)^{\top}\lambda^i_t) \\
                \star & 2\gamma^i_t-(\lambda^i_t)^{\top}\big(2h_i(\tilde{s}_t)+ G_i(\tilde{s}_t) \hat{u}^{-i}_t\big)
            \end{bmatrix} \succeq 0, \nonumber \\ \label{con_lmi}
            &\forall \; t \in [T],\: i \in \mathcal{I}, \\
            &\frac{1}{2}\big(\mathbf{R}(\theta)+ \mathbf{R}(\theta)^\top\big) \succeq \mu I \\
            &\lambda^i_t \geq 0, \: \forall \; t\in[T], \: i \in \mathcal{I}, \label{con_lamda}
        \end{align}
    \end{subequations}
    where $\hat{u}^{-i}_t := \sum_{j \in \mathcal{I}\setminus\{i\}} \tilde{u}^j_t$.
\end{prop}
The proof is deferred to Appendix \ref{proof_prop_qg}. Although the number of LMIs in \eqref{main_qg_prob} scales with the dataset size $T$ and the number of agents $N$, modern interior-point solvers are highly efficient at solving large-scale semidefinite programs. We now establish performance guarantees for the predictive capabilities of this framework. 

\subsection{Probabilistic performance guarantees}\label{pg_qg}
In this section, we derive statistical guarantees for the empirical loss minimization problem in \eqref{ig_main}. Specifically, we formalize bounds on the out-of-sample predictive capability of the learned forward game. Let $\hat{\theta}$ denote the optimal solution to \eqref{ig_main}, and consider an unseen testing dataset $\mathcal{D}^*:= \{(s^*_t,u^*_t)\}_{t=1}^{T'}$. By solving the $\hat{\theta}$-parameterized forward game \eqref{for_game} using dataset $\mathcal{D}^*$, we retrieve the predicted Nash equilibria $(\hat{u}_t)_{t=1}^{T'}$. We define the mean prediction error as
\begin{equation} \label{perf_criteria}
    \mathbf{e}_{T} = \frac{1}{T'} \sum_{t=1}^{T'} \lVert u^*_t-\hat{u}_t \rVert,
\end{equation}
and analyze the asymptotic behavior of $\mathbf{e}_T$ as the number of training data points $T \rightarrow \infty$. Let $\eta_T= \max_{t\in [T]} L_{\hat{\theta}}^{\mathrm{FO}}(\tilde{s}_t,\tilde{u}_t)$ denote the maximum FO loss evaluated at the parameter $\hat{\theta}$. Furthermore, let $F_{\theta}(s,u):= \mathrm{col}(\nabla_{u^i} J^i_\theta(u^i,u^{-i} \mid s))_{i\in \mathcal{I}}$ be the pseudogradient operator. We then state the following result. 
\begin{theorem}[Prediction error bound]\label{qg_theorem}
    Assume the operator $F_{\theta}$ is strongly monotone with parameter $\zeta$, and that Assumption \ref{feas_ass} holds. Then, for any $\alpha\in(0,1)$, with probability at least $1- \beta(\alpha)$ with respect to the sampling of the dataset, 
    \begin{equation} \label{mpe_bound}
        \mathbb{P}\left[ \mathbf{e}_T \leq \sqrt{\frac{\eta_T}{\zeta}}  \;\right]\geq 1-\alpha,
    \end{equation}
    where $\beta(\alpha)= \sum^{\mathrm{dim}(\theta)}_{i=0} \binom{T}{i}\alpha^i (1-\alpha)^{T-i}$, and $\mathbf{e}_T$ is defined in \eqref{perf_criteria}. 
\end{theorem}

\textbf{Proof.} The overarching approach relies on recognizing \eqref{main_qg_prob} as an uncertain convex program, allowing us to formalize a probabilistic bound for violating the variational inequality of the forward game. Using the strong monotonicity of the pseudogradient operator, we connect this probabilistic violation to the true solution. We first establish the following intermediate result.
    \begin{lem}[Chance constraint bound] \label{sp_lemma}
        Under Assumptions \ref{feas_ass}, if $\{ \hat{\theta}, (\hat{\lambda}_t,\hat{\gamma}_t)_{t\in [T]}\}$ is the optimal solution of \eqref{main_qg_prob}, then for any $\alpha \in (0,1)$, with probability at least $1- \beta(\alpha)$ over the sample generation, 
        \begin{equation}\label{pb_eq}
            \mathbb{P}\Big[\langle F_{\hat{\theta}}(\mathbf{s,u}),x-\mathbf{u} \rangle \geq -\eta_T,\:\forall x \in \mathcal{U}(\mathbf{s}) \Big] \geq 1- \alpha,
        \end{equation}
        where $\beta(\alpha)= \sum^{\mathrm{dim}(\theta)}_{i=0} \binom{T}{i}\alpha^i (1-\alpha)^{T-i}$.
    \end{lem}
    
      \textbf{Proof.} Let us first define the feasibility set
        \begin{align*}
            g(\chi):=\;& \Big \{ \epsilon \in \mathbb{R},\theta \in\Omega ~\mid~ \exists \lambda,\gamma\; \mathrm{such\; that\;}\eqref{con_lmi}-\eqref{con_lamda} \text{ are satisfied while using } \chi\Big \}
        \end{align*}
        where the uncertain dataset $\mathcal{D}$ is mapped to the random variable $\chi$. Using $g(\chi)$, the convex program \eqref{main_qg_prob} can be equivalently expressed as the scenario program 
        \begin{align} \label{sample_prog}
            \min \Big\{ \lVert \epsilon\rVert : (\epsilon,\theta)\in \bigcap_{t=1}^{T} g(\chi_t) \Big\}
        \end{align}
        where $\chi_t$ is a realization of $\chi$. We generalize the scenario program \eqref{sample_prog} over the continuous support set $\Delta$ as $\min_{\epsilon,\theta} \|\epsilon\| \;\mathrm{s.t.} \;  (\epsilon,\theta)\in g(\chi), \; \forall \chi \in \Delta,$ and define the violation probability as $V(\epsilon_T,\hat{\theta}) := \mathbb{P}[\chi \in \Delta \mid (\epsilon_T,\hat{\theta})\notin g(\chi)],$ where $\epsilon_T$ and $\hat{\theta}$ denote the optimal value and solution of \eqref{sample_prog}, respectively. The authors in \cite{campi2008exact} established the distribution-free bound on the violation probability as $\mathbb{P}\Big[V(\epsilon_T,\hat{\theta})>\alpha\Big] \leq \beta(\alpha).$ From Lemma \ref{loss_dominance}, we establish that $\epsilon_T \leq \eta_T$ at $\hat{\theta}$. Following Definition \ref{fo_loss_defn}, the violation condition can be reformulated in terms of an $\eta_T$-approximate variational inequality: $\langle F_{\hat{\theta}}(\mathbf{s,u}),x-\mathbf{u} \rangle \geq -\eta_T,\:\forall x \in \mathcal{U}(\mathbf{s})$. Invoking the bound on violation probability yields~\eqref{pb_eq}. \hfill $\blacksquare$

        Leveraging this lemma, we invoke a fundamental property of the operator $F_{\theta}$ as stated in \cite[Theorem $4.1$]{pang1987posteriori}. 
        \begin{theorem}[Optimality gap] \label{mono_theorem}
            If $F_{\theta}$ is strongly monotone with parameter $\zeta$, then any $\epsilon$-approximate solution $y$ of the VI($F_\theta,\mathcal{U}$) is bounded as $\lVert y-x^\star \rVert_2 \leq\sqrt{\frac{\epsilon}{\zeta}} $, where $x^\star$ is the optimal solution of the VI($F_\theta,\mathcal{U}$). 
        \end{theorem}
        Applying Theorem \ref{mono_theorem} to the $\eta_T$-approximate variational inequality in \eqref{pb_eq}, we substitute $(\mathbf{s,u})$ with a given testing data point $(s^*_j,u^*_j)$. By contradiction, if $\hat{u}_j$ diverges from $u^*_j$ by more than $\sqrt{\eta_T/\zeta}$, the event $\epsilon_T > \eta_T$ must have occurred, which Lemma \ref{sp_lemma} bounds with probability $\beta(\alpha)$. Extending this logic across all $T'$ testing points directly upper bounds the mean prediction error ($\mathbf{e}_T$) as stated in \eqref{mpe_bound}.      \hfill $\blacksquare$

It is noteworthy that $\eta_T$ in \eqref{mpe_bound} can turn out to be zero in the cases where the collected data is pure (without any noise), and/or the data is generated from consistent functions, i.e., the data generating true $\theta^\star \in \Omega$. Otherwise, any other sources of inconsistencies in the training data accumulate in $\eta_T$. It is also interesting to note that the probabilistic bound as in~\eqref{mpe_bound} is indeed tight even when $\eta_T=0$ (i.e., the training loss is zero); see~\cite{fatemi_nips} for a recent work in the context of inverse optimization, which is a special case of the games studied in this work.   
\section{Inverse Formulation for Non-Quadratic Monotone Games} \label{non_quad_sec}
Expanding the hypothesis space to capture more general functions is a practical imperative for learning complex forward games. However, one must balance expressivity with the necessity of a tractable inverse formulation. Furthermore, to ensure compatibility with popular Nash equilibrium-seeking algorithms~\cite{salehisadaghiani2016distributed}, we require the pseudogradient operator of the learned game to remain (strongly) monotone. To achieve this, we synthesize monotonicity constraints with nonparametric structures drawn from infinite-dimensional reproducing kernel Hilbert spaces (RKHS).

\subsection{Proposed formulation}
We consider the following foundational concept from \cite[Equation $1.3.6$]{chorin1990mathematical}. 
\begin{defn}[Helmholtz-Hodge decomposition]
Let $F$ be a sufficiently smooth vector field on a bounded domain $\Xi$ with a smooth boundary $\partial \Xi$. Then, $F$ can be uniquely decomposed as
\begin{equation} \label{hhd_main}
    F= \nabla\phi + r,
\end{equation}
where $\phi$ is a scalar potential function, and the vector field $r$ is divergence-free ($\nabla\cdot r=0$) and tangential to the boundary $\partial \Xi$. 
\end{defn}
The Helmholtz-Hodge decomposition is a cornerstone of vector calculus widely applied in fluid dynamics and geophysics \cite{bhatia2012helmholtz}. Inspired by this, we model the pseudogradient operator $F_{\theta}$ in a structurally analogous form
\begin{equation} \label{proposed_operator}
    F_{\theta}(s,u)= \nabla_u \phi_{\theta}(u) + \varphi(s)S_{\theta}u,
\end{equation}
where the scalar potential $\phi_{\theta}(u)$ is restricted to strictly convex functions, $\varphi:\mathbb{S}\rightarrow\mathbb{R}$ is a known feature map, and $S_{\theta}$ is constrained to be a skew-symmetric matrix (i.e., $S_{\theta}=-S_{\theta}^{\top}$). Here, the divergence-free vector field $r$ in \eqref{hhd_main} is replaced by $\varphi(s)S_{\theta}u$, satisfying $\nabla_u\cdot(\varphi(s)S_{\theta}u)=0$. Let us first gain an intuitive understanding of the proposed pseudogradient operator decomposition through the following example. 
\begin{exmp}[Helmholtz--Hodge decomposition] \label{hh_decom}
Consider the two-player game $J^1=u^1(0.35u^1-su^2)$ and $J^2=u^2(0.8u^2+su^1)$ with $s=1.1$. Its pseudogradient $ F(s,u)=[\,0.7u^1-1.1u^2,\;1.6u^2+1.1u^1\,]^\top $ decomposes as $ \nabla_u\phi(u)=[\,0.7u^1,\;1.6u^2\,]^\top,\; \varphi(s)Su=[\,-1.1u^2,\;1.1u^1\,]^\top. $ Figure \ref{vector_field_fig} illustrates the decomposition. Although chosen for visualization, this structured example highlights how the rotational and irrotational components of $F_\theta$ enable richer objective classes. 
\end{exmp}
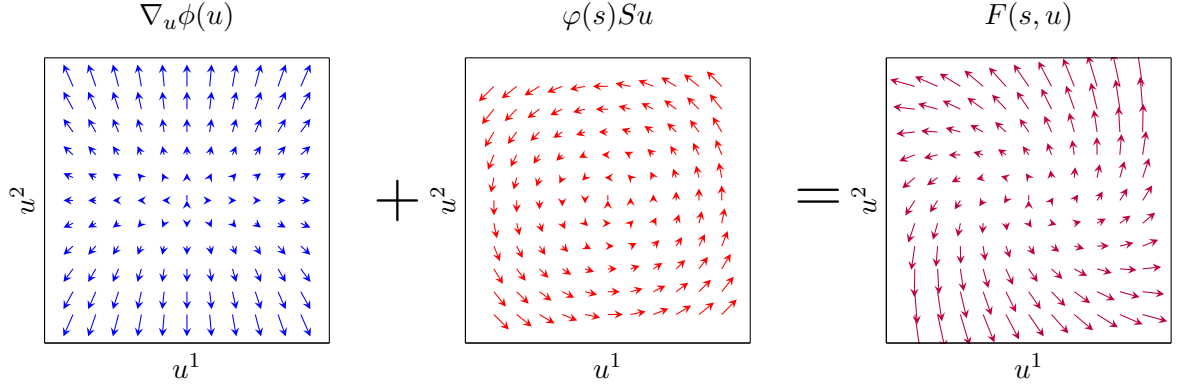
\begin{figure}[htbp]
    \centering
    \begin{tikzpicture}

        \begin{axis}[
            name=ax1,
            title={$ \nabla_u \phi(u)$},
            xmin=-2.5, xmax=2.5,
            ymin=-2.5, ymax=2.5,
            view={0}{90},
            axis equal image,
            width=6.2cm,
            xlabel=$u^1$, ylabel=$u^2$,
            xtick=\empty, ytick=\empty,
        ]
            \addplot3 [
                blue, -stealth,
                domain=-2:2, domain y=-2:2,
                samples=11,
                quiver={
                    u={0.7*x},
                    v={1.6*y},
                    scale arrows=0.12,
                }
            ] (x, y, 0);
        \end{axis}

        \node at ($(ax1.east) + (0.9cm,0)$) {\Huge $+$};

        \begin{axis}[
            name=ax2,
            at={(ax1.east)},
            anchor=west,
            xshift=1.8cm,
            title={$\varphi(s)S u$},
            xmin=-2.5, xmax=2.5,
            ymin=-2.5, ymax=2.5,
            view={0}{90},
            axis equal image,
            width=6.2cm,
            xlabel=$u^1$, ylabel=$u^2$,
            xtick=\empty, ytick=\empty,
        ]
            \addplot3 [
                red, -stealth,
                domain=-2:2, domain y=-2:2,
                samples=11,
                quiver={
                    u={-1.1*y},
                    v={1.1*x},
                    scale arrows=0.12,
                }
            ] (x, y, 0);
        \end{axis}

        \node at ($(ax2.east) + (0.9cm,0)$) {\Huge $=$};

        \begin{axis}[
            name=ax3,
            at={(ax2.east)},
            anchor=west,
            xshift=1.8cm,
            title={$F(s,u)$},
            xmin=-2.5, xmax=2.5,
            ymin=-2.5, ymax=2.5,
            view={0}{90},
            axis equal image,
            width=6.2cm,
            xlabel=$u^1$, ylabel=$u^2$,
            xtick=\empty, ytick=\empty,
        ]
            \addplot3 [
                purple, -stealth,
                domain=-2:2, domain y=-2:2,
                samples=11,
                quiver={
                    u={0.7*x - 1.1*y},
                    v={1.6*y + 1.1*x},
                    scale arrows=0.12,
                }
            ] (x, y, 0);
        \end{axis}

    \end{tikzpicture}
    \caption{Helmholtz-Hodge decomposition of the vector field $F$ into its irrotational (gradient) component $\nabla_u \phi(u)$ and its solenoidal (rotational) component $\varphi(s)Su$.}
    \label{vector_field_fig}
\end{figure}

One of the main advantages of this construction is that the strong monotonicity of $F_{\theta}$ can be guaranteed intrinsically. Consider the Jacobian $\mathbf{J}(s,u)= \nabla_uF_{\theta}(s,u)= \nabla^2_u\phi_{\theta}(u)+ \varphi(s)S_{\theta}$. The strong monotonicity condition is enforced by constraining $\frac{1}{2}(\mathbf{J}(s,u)+\mathbf{J}^{\top}(s,u))= \nabla^2_u\phi_{\theta}(u)+ \frac{1}{2}\varphi(s)(S_{\theta}+S_{\theta}^T)= \nabla^2_u\phi_{\theta}(u) \succ 0.$ This distills the challenge to modeling $\phi_{\theta}$ such that it flexibly captures non-quadratic objectives. We propose parameterizing $\phi_{\theta}$ as $\phi_{\theta}(u)= \frac{1}{2}u^{\top}R_{\theta}u + g_{\theta}(u),$ where $R_{\theta}\succ 0$, and $g_{\theta} \in \mathcal{H}_k$ belongs to the RKHS defined by a suitable kernel $k(x,y)$. Since we are directly modeling the pseudogradient operator, the empirical loss minimization is naturally formulated in terms of the FO loss \eqref{fo_loss}. We define the feasible operator class as $\Omega_F:= \{R_{\theta},S_{\theta} \mid  R_{\theta}= \theta_{R} I_n,\; 0 < \theta_R \leq \lambda_R,S_{\theta}=-S^{\top}_{\theta} \},$ and formalize the inverse formulation of the non-quadratic monotone game as
\begin{align}\label{erm_kernel}
    \min_{\theta \in \Omega_F,\; g_\theta \in \mathcal{H}_k} \quad & \sum_{t=1}^T \; L^{\mathrm{FO}}_{\theta}(\tilde{s}_t,\tilde{u}_t) + \lambda_{\mathrm{reg}} \lVert \nabla_u^2g_\theta \rVert^2_{\mathcal{H}_k} , \nonumber \\ 
    \mathrm{s.t.}\qquad & \nabla^2_u \phi_{\theta}(u)\succeq \lambda_{g}I_n,\;\; \forall u \in \mathcal{U}(\tilde{s}_t), \; t \in [T],
\end{align}
for a chosen margin $\lambda_g>0$, and regularizer penalty $\lambda_{\mathrm{reg}}>0$. The infinite-dimensional nature of $\phi_{\theta}$ obfuscates the tractability of \eqref{erm_kernel}. The following result guarantees a finite-dimensional convex reduction. 

\begin{theorem}[Tractable convex reformulation] \label{kernel_formulation}
    Consider the optimization problem \eqref{erm_kernel} where $F_\theta$ in~\eqref{proposed_operator} influencing the term $L^\mathrm{FO}_\theta$ in the objective (defined as in~\eqref{fo_loss}), Assumption \ref{feas_ass} holds, and the RKHS $\mathcal{H}_k$ is induced by the squared exponential kernel $k(x,y)= \exp(-\gamma_k ||x-y||^2_2)$.  A convex reformulation of the program \eqref{erm_kernel} is given by
    \begin{subequations} \label{erm_kernel_prop}
        \begin{align}
            \min_{\substack{ \theta \in \Omega_F, \epsilon\geq 0 \\ \lambda\geq 0, \alpha_\theta \in \mathbb{R}^T}} \;\; & ||\epsilon|| + \lambda_{\mathrm{reg}}(\alpha_\theta^{\top}K\alpha_\theta), \\
            \mathrm{s.t.} \qquad & F_{\theta}(\tilde{s}_t,\tilde{u}_t)= \theta_R I_n\tilde{u}_t + \sum_{j=1}^T \alpha_{\theta}^j D_{jt} \nonumber \\ 
            & + \varphi(\tilde{s}_t)S_{\theta}\tilde{u}_t,\quad \forall t \in [T],\\
            & G(\tilde{s}_t)^{\top}\lambda_t+F_{\theta}(\tilde{s}_t,\tilde{u}_t)=0,\; \forall t \in [T], \\
            &\tilde{u}_t^{\top}F_{\theta}(\tilde{s}_t,\tilde{u}_t)+h^{\top}(\tilde{s}_t)\lambda_t \leq \epsilon_t, \; \forall t \in [T],\\
            &(\theta_R-\lambda_g) \geq C_k \sqrt{\alpha_{\theta}^{\top}K_g\alpha_{\theta}},
        \end{align}
    \end{subequations}
    where $K$ is the Gram matrix with elementwise component as $[K]_{ij}= \nabla_{xy} k(x,y)\mid_{(\tilde{u}_i,\tilde{u}_j)}$, $K_g$ is Gram matrix corresponding to kernel $k(x,y)$, $D_{jt}=2\gamma_k(\tilde{u}_j-\tilde{u}_t)k(\tilde{u}_j,\tilde{u}_t) $, and $C_k=2\sqrt{3}\gamma_k$ is a kernel-specific constant. 
\end{theorem}
The proof is provided in Appendix \ref{proof_prop_kernel}. It is noteworthy that \eqref{main_qg_prob} is an exact reformulation of \eqref{ig_main}, whereas \eqref{erm_kernel_prop} is an approximation of \eqref{erm_kernel}. In addition, we learn the objective functions through \eqref{main_qg_prob}, whereas we learn pseudogradient operators by solving \eqref{erm_kernel_prop}.  
\subsection{Probabilistic performance guarantees}
We now derive probabilistic bounds on the prediction capabilities of the learned non-quadratic monotone game. We analyze the mean prediction error $\mathbf{e}_T$ \eqref{perf_criteria}, parameterized by $\hat{\theta}$ retrieved from \eqref{erm_kernel_prop}. Let $\hat{\eta}_{T}$ denote the optimal objective value.
\begin{theorem}[Out-of-sample performance] \label{nqg_theorem}
    Assume $F_{\theta}$ is strongly monotone with parameter $\zeta$, the feasible set $\mathcal{U}(\mathbf{s})$ is contained within a ball of radius $V$ almost surely, and Assumption \ref{feas_ass} holds. For any $\delta \in (0,1)$ and $\alpha>0$, with probability at least $1-\delta$ over the sampling of $\mathcal{D}$: 
    \begin{align} \label{mpe_kernel_bound}
        \mathbb{P}\left[ \mathbf{e}_T \leq\sqrt{\frac{\hat{\eta}_T + \alpha}{\zeta}}  \right]\geq 1-\Big(6GV \sqrt{\frac{\log (2/\delta)}{2T\alpha^2}}+  
        \frac{4(G+2MV)(\lambda_R-\lambda_g)}{C_kT\alpha} \sqrt{\mathrm{Tr}(K)}\Big),
    \end{align}
    where $A_\theta:=\theta_R I_n+ \varphi(s) S_\theta$, $ M= \|A_\theta\| +4\gamma_k C_\gamma \sum_{j=1}^T|\alpha_\theta^j|$, and $ G = \|A_\theta\|V + \sqrt{\frac{2\gamma_k}{e}} \sum_{j=1}^T|\alpha_\theta^j|$, for a constant \(C_\gamma>0\) dependent strictly on \(\gamma_k\).
\end{theorem}

\textbf{Proof.} Similar to the proof of Theorem \ref{qg_theorem}, we first want to formalize a generalization bound for the empirical loss function in \eqref{erm_kernel}. However, we have modeled the FO loss using a class of functions from an RKHS, an infinite-dimensional function space. Here, Rademacher complexity is used to measure the complexity of this class of functions. Let us first define the function class:
\begin{align} \label{main_fc}
    \mathfrak{H}:= \{(\Tilde{s},\Tilde{u})\mapsto \xi(f,\tilde{s},\tilde{u})\;|\; f\in \mathfrak{F}\}
\end{align}
where $\xi(\cdot)$ represents the FO loss in \eqref{fo_loss}, and the function class $\mathfrak{F}$ is defined as $\mathfrak{F}:=\{(\Tilde{s},\Tilde{u})\mapsto f(\theta,\tilde{s},\tilde{u})\;|\; \theta \in \Omega_F, \nabla^2_u \phi_{\theta}(\Tilde{u}_t)\succeq \lambda_{g}I_n,\; \forall t \in [T] \}$, where $f(\cdot)$ represents the pseudogradient operator. Let us further denote the random variable $(\tilde{s},\Tilde{u})=z$, and $S:= (z_1,z_2,\dots z_T)$ be the i.i.d sample. Consider the range of every function belonging to $\mathfrak{H}$ is contained in $[0,V_m]$. Then, the empirical Rademacher complexity for $\mathfrak{H}$ is given by $\hat{R}_T(\mathfrak{H})= \mathbb{E}_{\epsilon}[\sup_{h \in \mathfrak{H}} \frac{1}{T}\sum_{i=1}^T \epsilon_ih(z_i)]$, where $h \in \mathfrak{H}$, $\epsilon_i \in \{-1,1\}$ are independent random variables. Using $\hat{R}_T$, we want to obtain a generalization bound for $h(z)$. To this end, let $R(h) = \mathbb{E}[h(z)]$ and $\hat{R}_S(h) = \frac{1}{T} \sum_{i=1}^T h(z_i)$. Define the supremum of the generalization gap as $\Phi(S) = \sup_{h \in H} \left( R(h) - \hat{R}_S(h) \right)$. We proceed in four steps.

\textit{(i) Concentration of $\Phi(S)$:} 
Changing a single sample $z_i$ to $z_i'$ alters $\hat{R}_S(h)$ by at most $V_m/T$. Thus, the supremum $\Phi(S)$ changes by at most $V_m/T$. By McDiarmid's inequality, with probability at least $1-\delta/2$
\begin{equation} \label{eq:mcdiarmid_phi}
\Phi(S) \leq \mathbb{E}_S[\Phi(S)] + V_m \sqrt{\frac{\log(2/\delta)}{2T}}
\end{equation}
\textit{(ii) Symmetrization:} 
Introduce a ghost sample $\tilde{S} = \{\tilde{z}_1, \dots, \tilde{z}_T\}$, and using Jensen's inequality and the subadditivity of the supremum 
\begin{align*}
\mathbb{E}_S[\Phi(S)] &= \mathbb{E}_S \left[ \sup_{h \in H} \left( \mathbb{E}_{\tilde{S}}[\hat{R}_{\tilde{S}}(h)] - \hat{R}_S(h) \right) \right] \\
&\leq \mathbb{E}_{S, \tilde{S}, \epsilon} \left[ \sup_{h \in H} \frac{1}{T} \sum_{i=1}^T \epsilon_i (h(\tilde{z}_i) - h(z_i)) \right]
\leq 2 \mathbb{E}_{S, \epsilon} \left[ \sup_{h \in H} \frac{1}{T} \sum_{i=1}^T \epsilon_i h(z_i) \right] = 2 R_T(\mathfrak{H})
\end{align*}
where $R_T(\mathfrak{H}) = \mathbb{E}_S[\hat{R}_T(\mathfrak{H})]$ is the expected Rademacher complexity.

\textit{(iii) Concentration of $\hat{R}_T(\mathfrak{H})$:} 
Changing one point in $S$ alters the empirical Rademacher complexity $\hat{R}_T(\mathfrak{H})$ by at most $V_m/T$. Applying McDiarmid's inequality again, with probability at least $1-\delta/2$
\begin{equation} \label{eq:mcdiarmid_rademacher}
R_T(\mathfrak{H}) \leq \hat{R}_T(\mathfrak{H}) + V_m \sqrt{\frac{\log(2/\delta)}{2T}}
\end{equation}
\textit{(iv) Union Bound:} 
By the union bound, equations \eqref{eq:mcdiarmid_phi} and \eqref{eq:mcdiarmid_rademacher} hold simultaneously with probability at least $1-\delta$. Substituting \eqref{eq:mcdiarmid_rademacher} into the symmetrization bound, and then into \eqref{eq:mcdiarmid_phi} $\Phi(S) \leq 2 \left( \hat{R}_T(\mathfrak{H}) + V_m \sqrt{\frac{\log(2/\delta)}{2T}} \right) + V_m\sqrt{\frac{\log(2/\delta)}{2T}}= 2 \hat{R}_T(\mathfrak{H}) + 3V_m\sqrt{\frac{\log(2/\delta)}{2T}}$. Since $\Phi(S) = \sup_{h \in \mathfrak{H}} ( R(h) - \hat{R}_S(h) )$, rearranging the inequality, for any $\delta \in (0,1)$, with probability at least $1-\delta$ over an i.i.d. sample $S$, the following holds for all $h \in \mathfrak{H}$:
\begin{align}\label{exp_bound}
    \mathbb{E}[h(z)] \leq \frac{1}{T} \sum_{i=1}^T h(z_i) + 2 \hat{R}_T(\mathfrak{H}) + 3V_m \sqrt{\frac{\log (2/\delta)}{2T}}.
\end{align}
Now we upper bound $\hat{R}_T(\mathfrak{H})$ using Lemma~\ref{ub_H}. Note that $D\leq2V$ by initial assumption in Theorem \ref{nqg_theorem}, and from the proof of Lemma~\ref{ub_H}, we know that $L^{\mathrm{FO}}_\theta\leq 2GV$, which implies we can replace $V_m= 2GV$. Finally, replacing the result of Lemma \ref{ub_H} in \eqref{exp_bound}, and using Assumption~\ref{feas_ass}, 
\begin{align} \label{exp_bd}
    \mathbb{E}[\xi(\hat{\theta},\mathbf{s,u})] \leq \hat{\eta}_T   + \frac{4(G+2MV)(\lambda_R-\lambda_g)}{C_kT} \sqrt{\mathrm{Tr}(K)} + 6GV \sqrt{\frac{\log (2/\delta)}{2T}}
\end{align}
Notice that, conditional on the sample, the event $\mathbf{u}$ is not a $\hat{\eta}_T+\alpha$-approximate solution of VI is equivalent to the event $\xi_T > \hat{\eta}_T+\alpha$. Using Markov's inequality, we can infer from \eqref{exp_bd} that for any $\delta \in (0,1)$, with probability at least $1-\delta$ over an i.i.d. sample $S$
\begin{align}\label{eq_pb}
    \mathbb{P}\Big[\langle F_{\hat{\theta}}(\mathbf{s,u}),x-\mathbf{u} \rangle \geq -(\hat{\eta}_T+\alpha),\;\forall x \in \mathcal{U}(\mathbf{s}) \Big ] \geq 1 - \bigg(\frac{4(G+2MV)(\lambda_R-\lambda_g)}{C_kT\alpha} \sqrt{\mathrm{Tr}(K)} 
    \\ \nonumber
    + \, 6GV \sqrt{\frac{\log (2/\delta)}{2T\alpha^2}}\bigg )
\end{align}
In \eqref{eq_pb}, we replace $(\mathbf{s,u})$ with a testing datapoint $(s^*_j,u^*_j)$, and using contradiction if $\hat{u}_j$ is not $\sqrt{(\hat{\eta}_T+\alpha)/\zeta}$-near to $u^*_j$, then Theorem \ref{mono_theorem} states the event $\epsilon_T > (\hat{\eta}_T+\alpha)$ occurs with probability $\delta$. Extending the same argument for all $T'$ testing datapoints, the mean prediction error ($\mathbf{e}_T$) is bounded as \eqref{mpe_kernel_bound}. This concludes the proof. \hfill $\blacksquare$
\begin{rem}[Bounds comparison]
    We compare our bound with a similar bound in \cite[Theorem $8$]{bertsimas2015data}. We specifically compare the part $\frac{4(G+2MV)(\theta^*_R-\lambda_g)}{C_kT\alpha} \sqrt{\mathrm{Tr}(K)}$ in \eqref{mpe_kernel_bound} with the part $\frac{4p\overline{B}^p}{\alpha^p \sqrt{T} }$ in \cite[Theorem $8$]{bertsimas2015data}. For square-exponential kernel, $\mathrm{Tr}(K)=\mathcal{O}(T)$, which confirms that both bounds achieve an overall rate of $\mathcal{O}(1/\sqrt{T})$.  
\end{rem}
\section{Inverse formulation for linear-Quadratic dynamic games} \label{dg_game_section}
Here, we extend our inverse game formulation from a static to a dynamic setup. In particular, we consider the finite-horizon open-loop Nash equilibrium of linear-quadratic dynamic games.  
\subsection{Finite-horizon linear-quadratic dynamic game}
Let us consider a linear time-invariant (LTI) dynamical system of the form
\begin{align}\label{LTI}
    x_{t+1}= Ax_t + \sum_{i \in \mathcal{I}} B^iu^i_t,
\end{align}
which is governed by inputs from self-interested agents indexed by $i \in \mathcal{I}$, and where $A\in \mathbb{R}^{n_x\times n_x}$ is the state matrix, and $B^i\in \mathbb{R}^{n_x\times n_u^i}$ is the individual input matrix. Without loss of generality, consider for any $i,j\in \mathcal{I}$, $n^i_u=n^j_u$. The state $x_t \in \mathbb{R}^{n_x}$ is global and fully observable to all agents, whereas $u^i_t \in \mathbb{R}^{n_u^i}$ denotes the local control input from agent $i$. Each agent $i$ optimizes a local quadratic cost function defined over a finite horizon $T$:
\begin{align}\label{obj_unknown}
    J^i(x_t,u^i_t,u^{-i}_t) :=\sum_{k=0}^{T-1} \frac{1}{2}\Big[ x^{\top}_k Q^ix_k+ (q^i)^{\top}x_k +(u^i_k)^{\top}R^iu^i_k \Big] 
\end{align}
where $Q^i\in \mathbb{S}^{n_x}_{\succ 0}$ and $R^i\in \mathbb{S}^{n^i_u}_{\succ 0}$ are symmetric positive-definite penalty matrices for the states and local inputs, respectively.
In addition, the control input sequence for each agent must satisfy local constraints such that $u^i_k \in \mathcal{U}^i(s), \forall k \in \mathcal{T}:= \{0,1,\dots T-1\}$, $s \in \mathbb{S}$. Compared to \cite{benenati2025linear}, here we consider a standard formulation without an augmented state space. We formally define the finite-horizon linear-quadratic game \cite{bacsar1998dynamic,hall2025stability} as a forward problem ($\mathcal{P}_F(x_0)$) parameterized by the initial state $x_0$:
\begin{align} \label{rec_game}
    \forall i \in \mathcal{I}: \begin{cases}
         \min\limits_{u^i} &    J^i(x_0,u^i,u^{-i}), \\
        \:\text{s.t.}   & x_{k+1}= Ax_k + \sum_{i \in \mathcal{I}} B^iu^i_k, \forall k \in \mathcal{T},\\
         &u^i_k \in \mathcal{U}^i(s), \:\forall k \in \mathcal{T},
    \end{cases}
\end{align}
A standard methodology for determining the finite-horizon open-loop Nash equilibrium of the above problem is to reformulate it as a Variational Inequality (VI) problem \cite{belgioioso2022distributed}, such that the optimal joint action sequence ($u^*$) satisfies $ F(u^*;x_0) + \mathcal{N}_{\mathcal{U}}(u^*) \ni 0,$ where the pseudogradient operator is $F(u;x_0):= \text{col}(\nabla_{u^i} J^i(x_0,u^i,u^{-i}))_{i\in \mathcal{I}}$, and $\mathcal{N}_{\mathcal{U}}$ represents the normal cone of the collective decision set $\mathcal{U}(s):= \prod_{i\in \mathcal{I}} \mathcal{U}^i(s)$. In the context of the linear-quadratic dynamic games, the operator $F(u;x_0)$ admits the following affine structure:
\begin{align} \label{operator_eq}
    F(u;x_0) = \left( \Tilde{R} +  \mathbf{M}_u \right)u +  \mathbf{M}_x x_0 + \frac{1}{2}\mathbf{m}_q,
\end{align}
where $\Tilde{R} = \text{diag}(I_T \otimes R^1, \dots, I_T \otimes R^N)$, and $(i,j)$-th sub-blocks (or $i$-th sub-vectors) of the block matrices $\mathbf{M}_u$, $\mathbf{M}_x$, and $\mathbf{m}_q$ are given by $[\mathbf{M}_u]_{ij} = (\Tilde{B}^i)^{\top}\Tilde{Q}^i \Tilde{B}^j, \; \forall i,j \in \mathcal{I},\;[\mathbf{M}_x]_{i} = (\Tilde{B}^i)^{\top} \Tilde{Q}^i \Tilde{A},\;[\mathbf{m}_q]_{i} = (\Tilde{B}^i)^{\top} \Tilde{q}^i,$ with auxiliary matrices defined as $\Tilde{Q}^i= \text{diag}(Q^i, \dots, Q^i)$ repeated $T+1$ times, $\Tilde{A}= [I,A,A^2, \dots , A^T]^{\top}$, $\Tilde{q}^i = \text{col}(q^i)_{k=0}^T$, and $\Tilde{B}^i$ formulated as the standard lower-block-triangular convolution matrix associated with $A$ and $B^i$. Specific to our forward problem $\mathcal{P}_F$ in \eqref{rec_game}, in addition to Assumption \ref{feas_ass}, let us impose the following standing assumption.
\begin{assum}[Monotone operator]\label{uniquness_assump} The operator $F(\cdot,x_0)$ in \eqref{operator_eq} is $\mu$-strongly monotone for any initial state $x_0$, for some $\mu >0$.
\end{assum}
Under Assumption \ref{uniquness_assump}, \cite{hall2025stability} studies the uniqueness of the repeated Nash equilibria \cite[Prop. $1$]{hall2025stability} and convergence of the resulting closed-loop system to an asymptotically stable equilibrium \cite[Theorem $1$]{hall2025stability}, if there exists a positive definite matrix $P \in \mathbb{S}^{n_x}_{\succ 0}$ and a scalar $\eta > 0$ satisfying a linear matrix inequality (LMI) condition \cite[Eq. $(21)$]{hall2025stability} dependent on the system dynamics and the pseudogradient operator. Similar to our previous analysis, we parameterize the pseudogradient operator as $F_{\theta}$, and define the hypothesis space: $\Omega_D:= \{Q^i_{\theta}\in \mathbb{S}^{n_x}_{\succ 0}, R^i_{\theta}\in \mathbb{S}^{n_u}_{\succ 0}, q^i_{\theta}  \in \mathbb{R}^{n_x} : \lambda_Q I_{n_x} \preceq Q^i_{\theta} \preceq \lambda'_Q I_{n_x}, \: \lambda_R I_{n_u} \preceq R^i_{\theta} \preceq \lambda'_R I_{n_u}, \: \forall i \in \mathcal{I} \}$, with strictly positive scalars $\lambda < \lambda'$ chosen by the designer. Furthermore, to ensure the learned $F_{\theta}$ is strongly monotone, we restrict its symmetric Jacobian, $\frac{1}{2}(\mathbf{J}_u + \mathbf{J}_u^{\top}) \succeq \mu I$. From \eqref{operator_eq}, the Jacobian evaluates to $\mathbf{J}_u = \nabla_u F_{\theta}(u,\boldsymbol{x}) = \Tilde{R}_{\theta} +  \mathbf{M}_u$. Consequently, the strong monotonicity constraint simplifies to the LMI $\Tilde{R}_{\theta} + \frac{1}{2}\Big(\mathbf{M}_u(\theta) + \mathbf{M}_u^{\top}(\theta)\Big) \succeq \mu I,$ where the $(i,j)$-th block of the symmetric sum $\mathbf{M}_u(\theta) + \mathbf{M}_u^{\top}(\theta)$ is explicitly evaluated as $(\Tilde{B}^i)^{\top}(\Tilde{Q}^i_{\theta}+\Tilde{Q}^j_{\theta}) \Tilde{B}^j$. With this, we propose the following empirical loss minimization problem as the inverse formulation. 
\begin{align}\label{erm_dg}
    \min_{\theta \in \Omega_D} \quad & \sum_{t=1}^T \; L^{\mathrm{FO}}_{\theta}(\Tilde{s}_t,\Tilde{u}_t), \\ \nonumber
    \mathrm{s.t.}\quad & \Tilde{R}_{\theta} + \frac{1}{2}\Big( \mathbf{M}_u(\theta) + \mathbf{M}_u^{\top}(\theta) \Big) \succeq \mu I.
\end{align}
We provide the following tractable convex reformulation of \eqref{erm_dg}. 
\begin{prop}[Exact reformulation] \label{dg_formulation}
    Consider the optimization problem \eqref{erm_dg} where the candidate operator $F_{\theta}$ admits \eqref{operator_eq}, and let Assumption \ref{feas_ass} hold. Then the optimization problem \eqref{erm_dg} is equivalent to the following convex program. 
    \begin{subequations} \label{erm_dg_prop}
        \begin{align}
            \min_{\substack{ \theta \in \Omega_D, \{\epsilon_t\geq 0, \\
            \lambda_t\geq 0\}_{t\in [T]}}} \quad & ||\epsilon||, \\
            \mathrm{s.t.} \qquad & F_{\theta}(\Tilde{s}_t,\Tilde{u}_t)= \left( \Tilde{R}  +  \mathbf{M}_u(\theta) \right)\tilde{u}_t  \\ \nonumber
            & +  \mathbf{M}_x \Tilde{s}_t + \frac{1}{2}\mathbf{m}_q(\theta) \;\; \; \forall t \in [T],\\
            & G(\Tilde{s}_t)^{\top}\lambda_t+F_{\theta}(\Tilde{s}_t,\Tilde{u}_t)=0,\; \forall t \in [T], \\
            &\Tilde{u}_t^{\top}F_{\theta}(\Tilde{s}_t,\Tilde{u}_t)+h^{\top}(\Tilde{s}_t)\lambda_t \leq \epsilon_t, \; \forall t \in [T],\\
            &\Tilde{R}_{\theta} + \frac{1}{2}\Big(\mathbf{M}_u(\theta) + \mathbf{M}_u^{\top}(\theta)\Big) \succeq \mu I,
        \end{align}
    \end{subequations} 
\end{prop}
\begin{pf}
    The proof follows the duality argument introduced in \eqref{duality_eq} in the proof of Theorem \ref{kernel_formulation} in Appendix~\ref{proof_prop_kernel}. \hfill $\blacksquare$
\end{pf}
Note that the initial state ($x_0$) is treated as a signal $\tilde{s}$. One could impose stability of the associated open-loop Nash equilibrium trajectories by imposing \cite[Eq. ($21$)]{hall2025stability}. In our numerical experience, however, this condition is conservative and does not lead to improved solutions. The performance guarantees follow the same line of argument as Theorem \ref{qg_theorem}.
\begin{cor}[Out-of-sample performance] \label{dg_perf_bound}
    Let Assumptions \ref{feas_ass} and \ref{uniquness_assump} hold, then for any $\alpha\in(0,1)$, with probability at least $1- \beta(\alpha)$ with respect to the sampling $\mathbb{P}\left[ \mathbf{e}_T \leq \sqrt{\frac{\overline{\eta}_T}{\mu}}  \right]\geq 1-\alpha$, where $\overline{\eta}_T$ is the optimal value of objective function in \eqref{erm_dg_prop}, and $\beta(\alpha)= \sum^{\mathrm{dim}(\theta)}_{i=0} \binom{T}{i}\alpha^i (1-\alpha)^{T-i}$.
\end{cor}
\begin{pf}
    The proof follows the same line as that of Theorem \ref{qg_theorem}. Specifically, we formalize a generalization bound as in Lemma \ref{sp_lemma}, then, using the strong monotonicity property of $F_{\theta}$, we invoke Theorem \ref{mono_theorem} to finalize the probabilistic performance bound. \hfill $\blacksquare$
\end{pf}
\begin{rem}[Limitations of the offline guarantee] 
Corollary \ref{dg_perf_bound} provides a distribution-free guarantee for i.i.d. episodic data, suited to batch settings with independent encounters. In continuous ego-centric deployments, however, sliding trajectory windows are temporally correlated through the dynamics, violating the i.i.d. assumption. Finite-sample guarantees for such dynamically correlated data remain an open problem. 
\end{rem}
\subsection{Robust receding horizon learning} \label{noise_section}
While the offline formulation provides a foundation for inverse dynamic games, sequential deployment requires robust learning under sensor and process noise. We therefore seek to update $F_\theta$ in a receding-horizon manner without unbounded error accumulation or growing computational complexity. At time $t$, the observer processes the sliding window $\mathcal{D}_t=\{(\tilde{s}_t,\tilde{u}_t)\}$, where $\tilde{u}_t=[\tilde{u}(t),\ldots,\tilde{u}(t-T+1)]$ and the oldest sample is discarded to maintain length $T$. Consequently, the number of LMIs per update remains constant over time, enabling tractable online implementation.
\begin{assum} [Gaussian measurement noise] \label{noise_assumption}
    The vectorized observation entering the dataset $\mathcal D_t$ is of the form $ y_t=\bar y_t+\varrho_t, \; \varrho_t\sim\mathcal{N}(0,\sigma^2 I), $ where $\bar y_t$ denotes the noise-free observation and the variance is strictly bounded by a known hardware or process constant $0<\sigma^2 \leq \sigma^2_m$. Moreover, the measurement noise is independent across time.
\end{assum}
To reject this noise while accurately tracking the underlying parameter drift, we introduce a robust Tikhonov-regularized online update scheme. To ensure that temporal regularization is imposed only on the physically meaningful game parameters, we separate the structural parameter $\theta$ from the auxiliary variables of the inverse game problem. Let $ \nu := \operatorname{col} \big( (\lambda^j,\epsilon^j)_{j\in[T]} \big) $ collect the dual and epigraph variables appearing in~\eqref{erm_dg_prop}. For each fixed $\theta$, we rewrite~\eqref{erm_dg_prop} in terms of the following profiled empirical loss 
\begin{equation} \label{eq:profiled-loss} 
\widetilde{\mathcal L}_t(\theta) := \inf_{\nu} \left\{ \widetilde J_t(\theta,\nu) \;\middle|\; (\theta,\nu)\in\mathcal Z(\mathcal D_t) \right\}, 
\end{equation} 
with the convention that $\widetilde{\mathcal L}_t(\theta)=+\infty$ whenever no feasible auxiliary variable $\nu$ exists, and the feasibility set $\mathcal{Z}$ enforces the constraints of \eqref{erm_dg_prop}. The domain of the profiled loss is therefore $ \Theta_t := \left\{ \theta \;\middle|\; \exists\,\nu: (\theta,\nu)\in\mathcal Z(\mathcal D_t) \right\}. $ Let 
\begin{equation} \label{eq:unregularized-profiled-estimator} 
\widehat\theta_t \in \arg\min_{\theta\in\Theta_t} \widetilde{\mathcal L}_t(\theta) 
\end{equation} 
denote the unregularized estimator. We propose the following Tikhonov-regularized update:  
\begin{align}\label{eq:regularized-profiled-problem}
\theta_t \in \arg\min_{\theta\in\Theta_t}\quad& \widetilde{\mathcal L}_t(\theta) + \rho_t \|\theta-\theta_{t-1}\|_{\mathrm{Fr}}^2, 
\end{align} 
where $\rho_t>0$ is an adaptive penalty. Equivalently, without explicitly evaluating the profiled loss, \eqref{eq:regularized-profiled-problem} can be implemented as 
\begin{subequations} \label{eq:regularized-full-problem} 
\begin{align} \min_{\theta,\nu}\quad& \widetilde J_t(\theta,\nu) + \rho_t \|\theta-\theta_{t-1}\|_{\mathrm{Fr}}^2, \label{eq:regularized-full-problem-a}\\ 
\text{s.t.}\quad& (\theta,\nu)\in\mathcal Z(\mathcal D_t). \label{eq:regularized-full-problem-b} 
\end{align} 
\end{subequations} 
Thus, the auxiliary variables are re-optimized at each time step but are not regularized across successive windows. While standard approaches treat $\rho_t$ as a static, heuristically tuned hyperparameter, we seek an adaptive weight that is optimal in the sense of Bayesian inference.
To formalize this, let us consider the following assumptions. 
\begin{assum}[Local regularity] \label{ass:local-MAP} 
Let $\theta_t^\star$ denote the minimizer of the noise-free profiled inverse problem associated with $\bar y_t$. The following conditions hold locally around $\theta_t^\star$ and $\widehat\theta_t$. 
\begin{enumerate}\renewcommand{\labelenumi}{(\roman{enumi})}
\item \emph{Strong regularity.} For the unregularized inverse problem in~\eqref{eq:unregularized-profiled-estimator}, the optimizer mapping $\mathcal A_t:y_t\twoheadrightarrow\widehat\theta_t $ is single-valued and continuously differentiable in a neighborhood of $\bar y_t$. 
\item \emph{Locally invariant active set.} For sufficiently small $\sigma_m$, the active constraint set at the unregularized inverse solution in~\eqref{eq:unregularized-profiled-estimator} does not change under perturbations of $y_t$ in a neighborhood of $\bar y_t$. 
\item \emph{Locally normalized profiled loss.} The profiled loss $\widetilde{\mathcal L}_t$ is twice continuously differentiable in a neighborhood of $\widehat\theta_t$, and its local curvature is normalized according to $ \nabla_\theta \widetilde{\mathcal L}_t(\widehat\theta_t) =0, \;\mathrm{and}\; \nabla_\theta^2 \widetilde{\mathcal L}_t(\widehat\theta_t) =I. $ 
\item \emph{Gaussian parameter evolution.} Conditional on the information available at time $t-1$, the true parameter follows $ \theta_t^\star = \theta_{t-1}+w_t, \; w_t\sim\mathcal N(0,\sigma_{s,t}^2I), $ where $\sigma_{s,t}^2>0$, and the process disturbance $w_t$ is conditionally independent with measurement noise $\varrho_t$. 
\end{enumerate}
\end{assum} 
Assumption~\ref{ass:local-MAP}(iii) is a local normalization condition rather than a global restriction on the empirical loss. More generally, if the profiled loss has a positive-definite reduced Hessian $ H_t := \nabla_\theta^2 \widetilde{\mathcal L}_t(\widehat\theta_t) \succ0, $ then Taylor's theorem gives 
\begin{align} \label{eq:local-loss-expansion}
    \widetilde{\mathcal L}_t(\theta) = \widetilde{\mathcal L}_t(\widehat\theta_t) + \frac{1}{2} \|\theta-\widehat\theta_t\|_{H_t}^2 + o\!\left( \|\theta-\widehat\theta_t\|_{\mathrm{Fr}}^2 \right).
\end{align}
Hence, the condition $\nabla_\theta^2\widetilde{\mathcal L}_t(\widehat\theta_t)=I$ can be understood as working in locally curvature-normalized parameter coordinates. Let $ S_t := D_{\overline{y}_t}\mathcal A_t(\bar y_t) $ denote the sensitivity of the unregularized estimator with respect to the vectorized observation, where $D_{\overline{y}_t}(\cdot)$ denotes the Jacobian operator. Under Assumption~\ref{ass:local-MAP}, a first-order perturbation gives $ \widehat\theta_t-\theta_t^\star = S_t\varrho_t + o_{\mathbb P}(\|\varrho_t\|). $ The exact local covariance of this perturbation is generally anisotropic and is given by $ C_t = \sigma_m^2S_tS_t^\top$. Since $S_tS_t^\top$ is symmetric, $C_t$ admits the exact decomposition into isotropic and traceless anisotropic components as $C_t= \sigma_{\mathrm{meas},t}^2I+ (\sigma_m^2S_tS_t^\top-\sigma_{\mathrm{meas},t}^2I),$ where $\sigma_{\mathrm{meas},t}^2 := \frac{\sigma_m^2}{d_\theta} \operatorname{Tr}(S_tS_t^\top)$, and $d_\theta:=\dim(\theta)$. We retain the isotropic component and neglect the traceless anisotropic component in order to obtain a scalar Tikhonov weight. Thus, we consider   
\begin{equation} \label{eq:scalar-measurement-variance} 
C_t= \sigma_{\mathrm{meas},t}^2I.  
\end{equation} 
The division by $d_\theta$ in $\sigma_{\mathrm{meas},t}^2$ ensures that the isotropic approximation preserves the total propagated variance: $\operatorname{Tr}(C_t) = d_\theta\sigma_{\mathrm{meas},t}^2. $ Furthermore, let us consider the maximum a posteriori (MAP) estimate of $\theta^\star_t$ is given by
\begin{align} \label{map_formula}
    \theta^{\mathrm{MAP}}_t = \arg\min_{\theta \in \Theta_t} \;-\log\Big(p(\theta\mid \hat{\theta}_t,\theta_{t-1})\Big),
\end{align}
where $p(\theta^\star\mid \hat{\theta}_t,\theta_{t-1})$ denotes the density function of the posterior distribution following Bayes' rule. 
\begin{prop}[Local MAP equivalence] \label{prop:MAP-equivalence} Suppose that Assumptions~\ref{noise_assumption} and \ref{ass:local-MAP} hold, and the propagated covariance of the unregularized estimator is given by~\eqref{eq:scalar-measurement-variance}. Then, in the small-noise regime $\sigma_m\rightarrow0$, the MAP of $\theta_t^\star$ in~\eqref{map_formula} is locally, up to second-order terms, equivalent to the solution of the Tikhonov-regularized problem \eqref{eq:regularized-profiled-problem}, provided that 
\begin{equation} \label{eq:adaptive-rho} 
\rho_t = \frac{ \sigma_{\mathrm{meas},t}^2 }{ 2\sigma_{s,t}^2 }. \end{equation} 
\end{prop} 
The proof is given in Appendix~\ref{map_prop_proof}. 
\begin{rem}[Scope of the MAP equivalence] \label{rem:MAP-scope} Proposition~\ref{prop:MAP-equivalence} establishes local asymptotic, not global, equivalence, relying on the first-order optimizer sensitivity, isotropic moment matching in \eqref{eq:scalar-measurement-variance}, and the local quadratic model \eqref{eq:local-loss-expansion}. Active-set changes or optimizer non-differentiability may invalidate the local Gaussian approximation; in this case, \eqref{eq:regularized-profiled-problem} remains a tractable convex estimator, but its MAP interpretation is only heuristic.
\end{rem} 

Analytically computing the mapped variance $\sigma_{\mathrm{meas},t}^2 =\frac{\sigma_m^2}{d_\theta} \mathrm{Tr}(S_t S_t^{\top})$ requires evaluating the full sensitivity matrix $S_t$ that analytically poses a severe computational bottleneck. To circumvent this, we utilize a randomized trace estimator based on Hutchinson's approximation \cite{hutchinson1989stochastic}. Let $\zeta \in \mathbb{R}^{\text{dim}(\overline{y}_t)}$ be a random perturbation vector drawn from the assumed noise distribution $\mathcal{N}(0, \sigma_m^2 I)$. Applying a first-order Taylor expansion around the unperturbed empirical minimizer $\hat{\theta}_t(\overline{y}_t)$, the solution to the perturbed problem is approximated as $\hat{\theta}_t(\overline{y}_t + \zeta) \approx \hat{\theta}_t(\overline{y}_t) + S_t \zeta$. The squared Frobenius deviation between the nominal and perturbed solutions is therefore $\lVert \hat{\theta}_t(\overline{y}_t + \zeta) - \hat{\theta}_t(\overline{y}_t) \rVert_{\text{Fr}}^2 \approx \zeta^{\top} S_t^{\top} S_t \zeta.$ Taking the expectation with respect to the isotropic noise distribution $\zeta$ yields $\mathbb{E}_{\zeta} \Big[ \zeta^{\top} S_t^{\top} S_t \zeta \Big] = \text{Tr}\Big(S_t^{\top} S_t \mathbb{E}_{\zeta}[\zeta \zeta^{\top}]\Big)= \sigma_m^2 \text{Tr}(S_t^{\top} S_t) = d_\theta\sigma_{\mathrm{meas},t}^2.$
Consequently, a single Monte Carlo realization provides an unbiased estimate of the required penalty variance. In practice, the observer evaluates $\sigma_{\text{meas},t}^2$ by solving the unregularized optimization problem exactly one additional time with the globally perturbed observation. This reduces the sensitivity mapping from $\mathcal{O}(\text{dim}(\overline{y}_t))$ to $\mathcal{O}(1)$ complexity, ensuring tractability of our online update. The complete proposed framework is illustrated in Fig. \ref{fig_complete_op}.
\begin{center}
    \includegraphics[
        width=0.9\textwidth,
        height=0.38\textheight,
        keepaspectratio
    ]{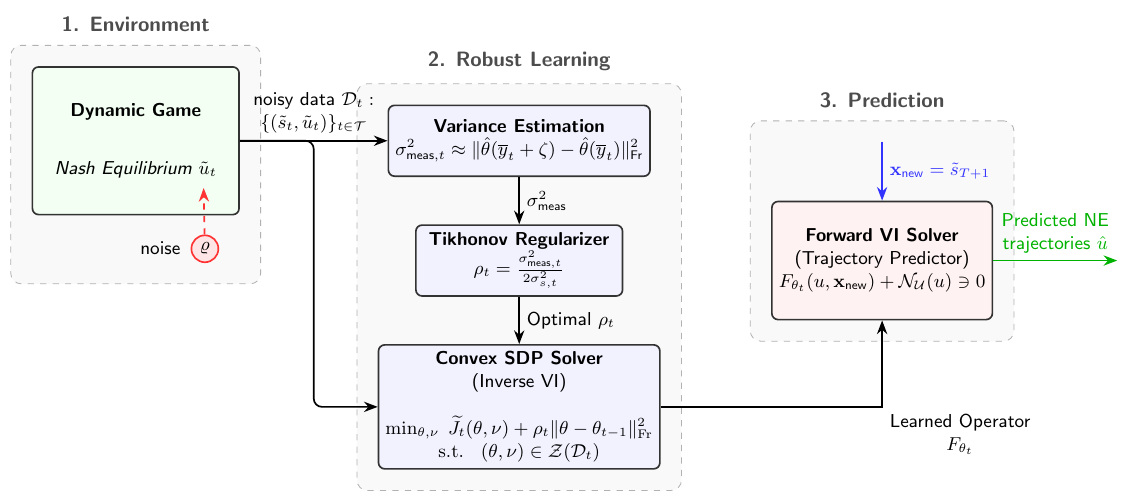}

    \captionof{figure}{Schematic of proposed robust receding horizon learning methodology: from observations, to learning the game operator, to predicting the Nash equilibrium trajectory. }
    \label{fig_complete_op}
\end{center}

\section{Numerical Experiments}\label{numerical_sec}
We present two numerical examples and a stylized application. Forward VIs for data generation and out-of-sample testing are solved using aGRAAL \cite{malitsky2020golden}, while the inverse SDPs are solved with MOSEK in MATLAB.
\subsection{Quadratic and non-quadratic monotone games}
We first establish the ground-truth forward games for both instances, considering an $N=2$ player setup with $n_i=3$ decision variables per player. For the quadratic game objective in \eqref{obj_qg}, the parameters $R^i$ and $r^i$ are generated randomly such that $(R^i, r^i) \in \Omega$ for all $i \in \mathcal{I}$. For the non-quadratic game, the objective function is $J^i(u^i,u^{-i} \mid s) := \frac{1}{2} (u^i)^{\top} \Big(R^i \sum_{j=1}^N u^j\Big) + s(r^i)^{\top} u^i + \|u^i\|^4_2, \; \forall i \in \mathcal{I},$ where $R^i$ and $r^i$ are similarly generated. Both games operate over randomly generated polytopic feasibility sets. To compile the training datasets, we compute the forward game equilibria for $T=50$ and $T=80$ samples for the quadratic and non-quadratic games, respectively. For out-of-sample validation, we evaluate the learned models over 10 independent test episodes, consisting of $T'=50$ testing samples for the quadratic game and $T'=80$ for the non-quadratic game.

\begin{figure}
    \centering
    \begin{subfigure}{0.47\linewidth}
        \centering
        \includegraphics[width=0.9\linewidth]{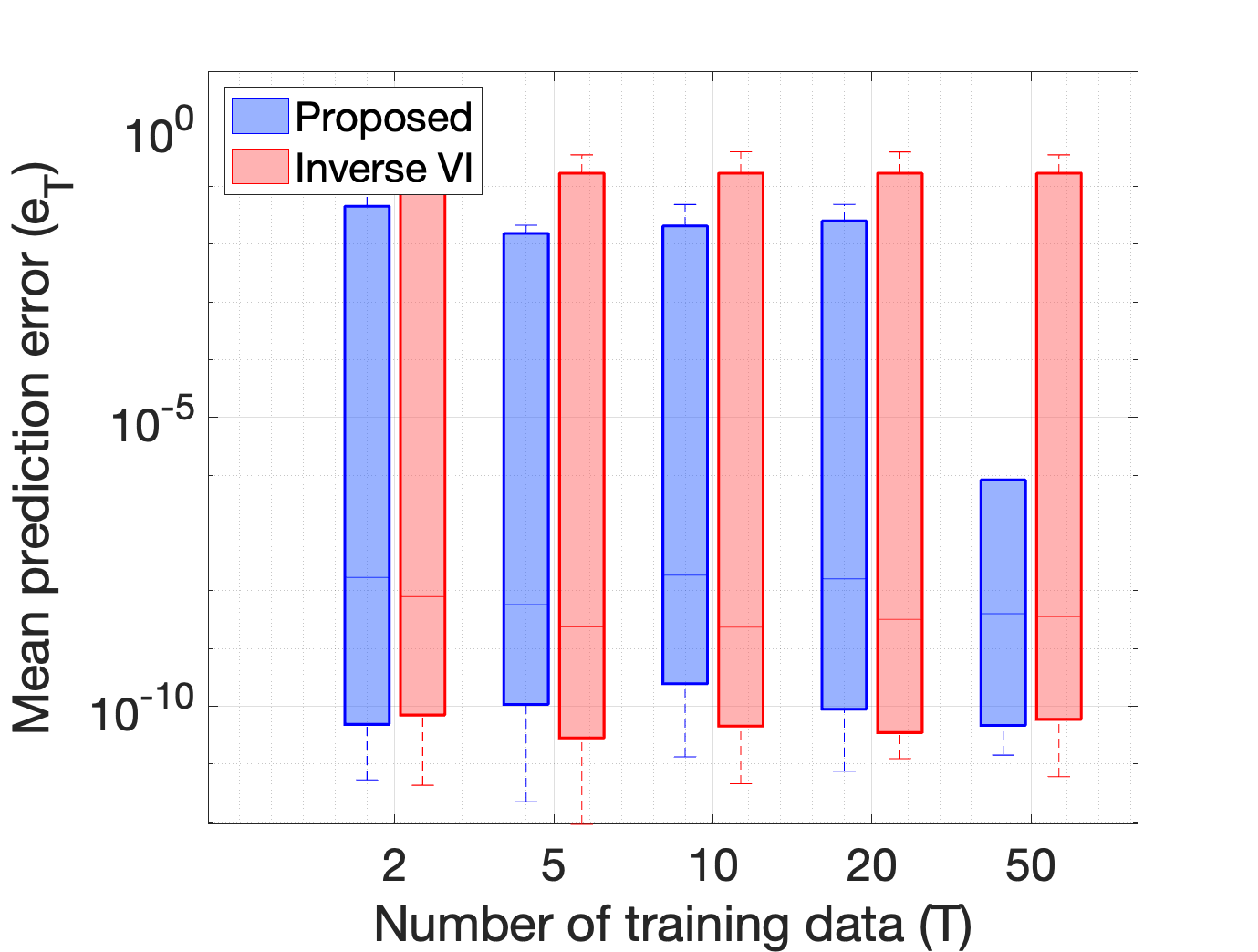}
        \caption{}
        \label{qg_fig}
    \end{subfigure}
    \begin{subfigure}{0.47\linewidth}
        \centering
        \includegraphics[width=0.9\linewidth]{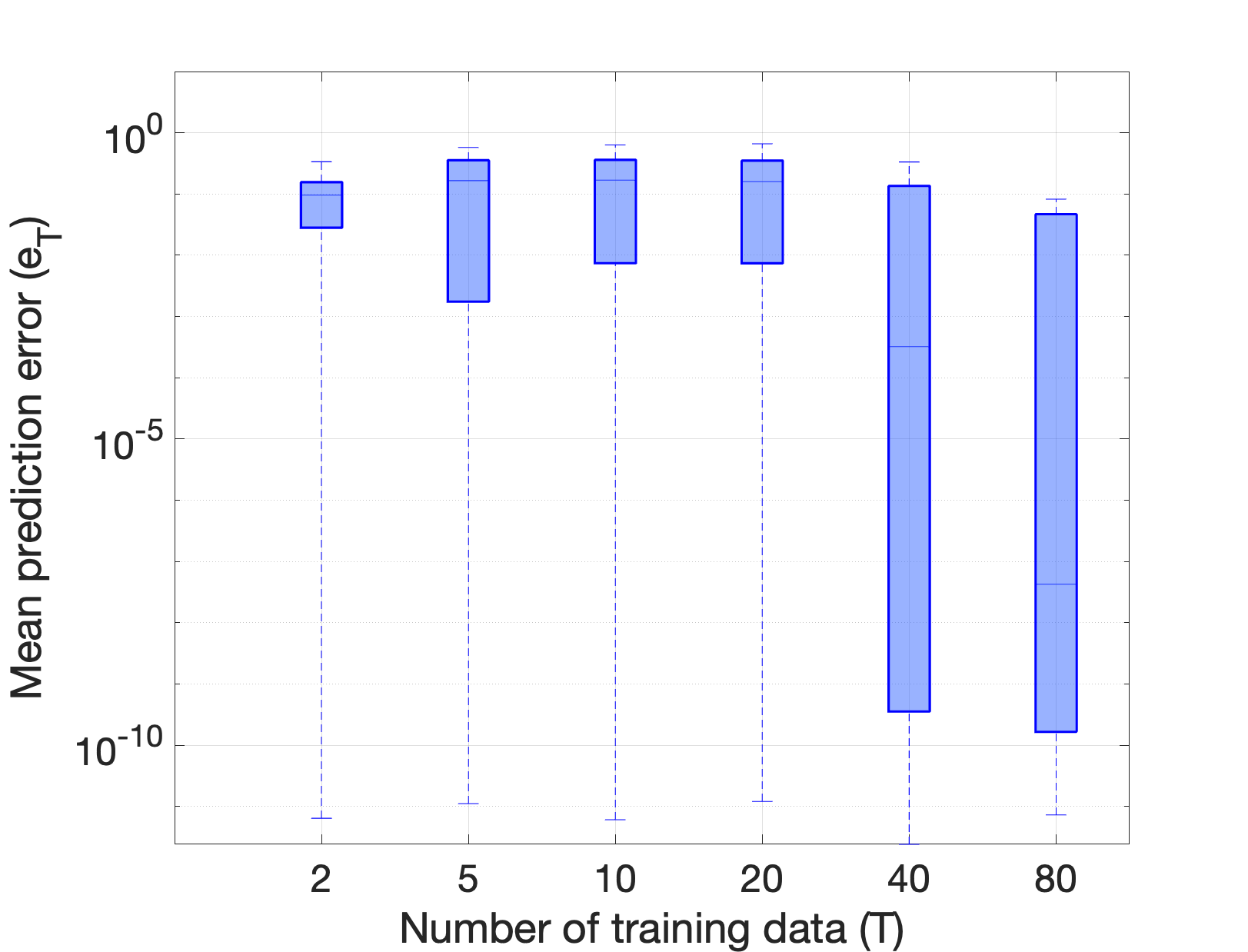} 
        \caption{}
        \label{nqg_fig}
    \end{subfigure}
    \caption{\small Out-of-sample performance of the learned game as a function of the training dataset size, evaluated over a fixed number of testing samples. (a) Comparison of the mean prediction error ($\mathbf{e}_T$ defined in \eqref{perf_criteria}) between the proposed NI loss-based method and the FO loss-based Inverse VI framework \cite{bertsimas2015data} for the quadratic game. (b) Mean prediction error for the non-quadratic monotone game.}
    \label{sim_res_fig}
\end{figure}

\textit{Discussion of results.} 
For the quadratic game setup, we benchmark our NI loss-based formulation against the FO loss-based Inverse VI approach proposed in \cite{bertsimas2015data}. Fig.~\ref{qg_fig} illustrates the decaying mean prediction error as the training dataset size $T$ increases, evaluated over a fixed validation set of $T'=50$. While the mean performance of both methods is comparable, the proposed NI formulation exhibits noticeably lower variance across the independent test episodes. Furthermore, Fig.~\ref{nqg_fig} demonstrates a sublinear decay rate for $\mathbf{e}_T$ with respect to the training sample size $T$ in the non-quadratic setting, providing strong empirical validation for the statistical bounds established in Theorem~\ref{nqg_theorem}.
\subsection{Stylized application: collision avoidance game}
We validate the linear-quadratic formulation using a 1D head-on collision avoidance game \cite{peters2023online}. Two autonomous vehicles (agents 1 and 2) navigate towards each other along a shared corridor. The joint state is $x_t = [p^1_t, v^1_t, p^2_t, v^2_t]^{\top} \in \mathbb{R}^4$, comprising positions and velocities. The control input $u^i_t \in \mathbb{R}$ is the longitudinal acceleration. For the discretized double-integrator model ($\Delta t=0.2$\,s), \eqref{LTI} has $ A=\left[\begin{smallmatrix} 1&0.2&0&0\\0&1&0&0\\0&0&1&0.2\\0&0&0&1 \end{smallmatrix}\right],\quad B^1=\left[\begin{smallmatrix}0.02\\0.2\\0\\0\end{smallmatrix}\right],\quad B^2=\left[\begin{smallmatrix}0\\0\\0.02\\0.2\end{smallmatrix}\right]. $
Control inputs are physically constrained to $\mathcal{U}^i := [-a_{\max}, a_{\max}]$ with $a_{\max} = 5.0 \text{ m/s}^2$, and the receding horizon is $T = 20$. Exact collision avoidance constraints are non-convex, destroying the strong monotonicity of the VI. To ensure a strictly convex game, we employ a surrogate safety cost. Agents track nominal velocities ($v_{\text{ref}}^i$) driving them toward the origin, penalized by deviations from a safe inter-vehicle distance ($D_s$). Assuming agent 1 is on the positive axis and agent 2 on the negative, their distance is $(p^1_k - p^2_k)$. The local costs are $J^1= \sum_{k=0}^{T-1} \frac{1}{2}\Big[w_v (v^1_k-v_{\text{ref}}^1)^2+w_s(p^1_k-p^2_k-D_s)^2+R^1 (u^1_k)^2\Big],J^2=\sum_{k=0}^{T-1} \frac{1}{2}\Big[w_v (v^2_kv_{\text{ref}}^2)^2+w_s (p^1_k-p^2_k-D_s)^2+R^2 (u^2_k)^2\Big].$ This couples the agents via the shared spatial penalty, forcing them to negotiate braking profiles. Casting this into the quadratic form \eqref{obj_unknown}, the ground-truth parameters are
\begin{align*} 
Q^i_{\rm true} &=2w_s \left[\begin{smallmatrix} 1&0&-1&0\\ 0&0&0&0\\ -1&0&1&0\\ 0&0&0&0 \end{smallmatrix}\right] +2w_v e_{2i}e_{2i}^\top+\epsilon I,\quad i=1,2,\\ q^i_{\rm true} &=2w_sD_s[-1,0,1,0]^\top -2w_vv_{\rm ref}^i e_{2i},\quad i=1,2, 
\end{align*}
where $e_j$ denotes the $j$-th canonical basis vector of $\mathbb{R}^4$, $\epsilon = 10^{-3}$ ensures strict positive definiteness. The parameters are $v_{\text{ref}}^1 = -2.5 \text{ m/s}$, $v_{\text{ref}}^2 = 2.5 \text{ m/s}$, $D_s = 5 \text{ m}$, $w_v = 10$, and $w_s = 2$. Control effort penalties $R^1_{\text{true}} = 20$ and $R^2_{\text{true}} = 40$ reflect differing vehicle inertias.

\begin{figure}[t] 
\centering 
\begin{subfigure}[t]{0.32\linewidth} 
\centering 
\includegraphics[width=1.02\linewidth]{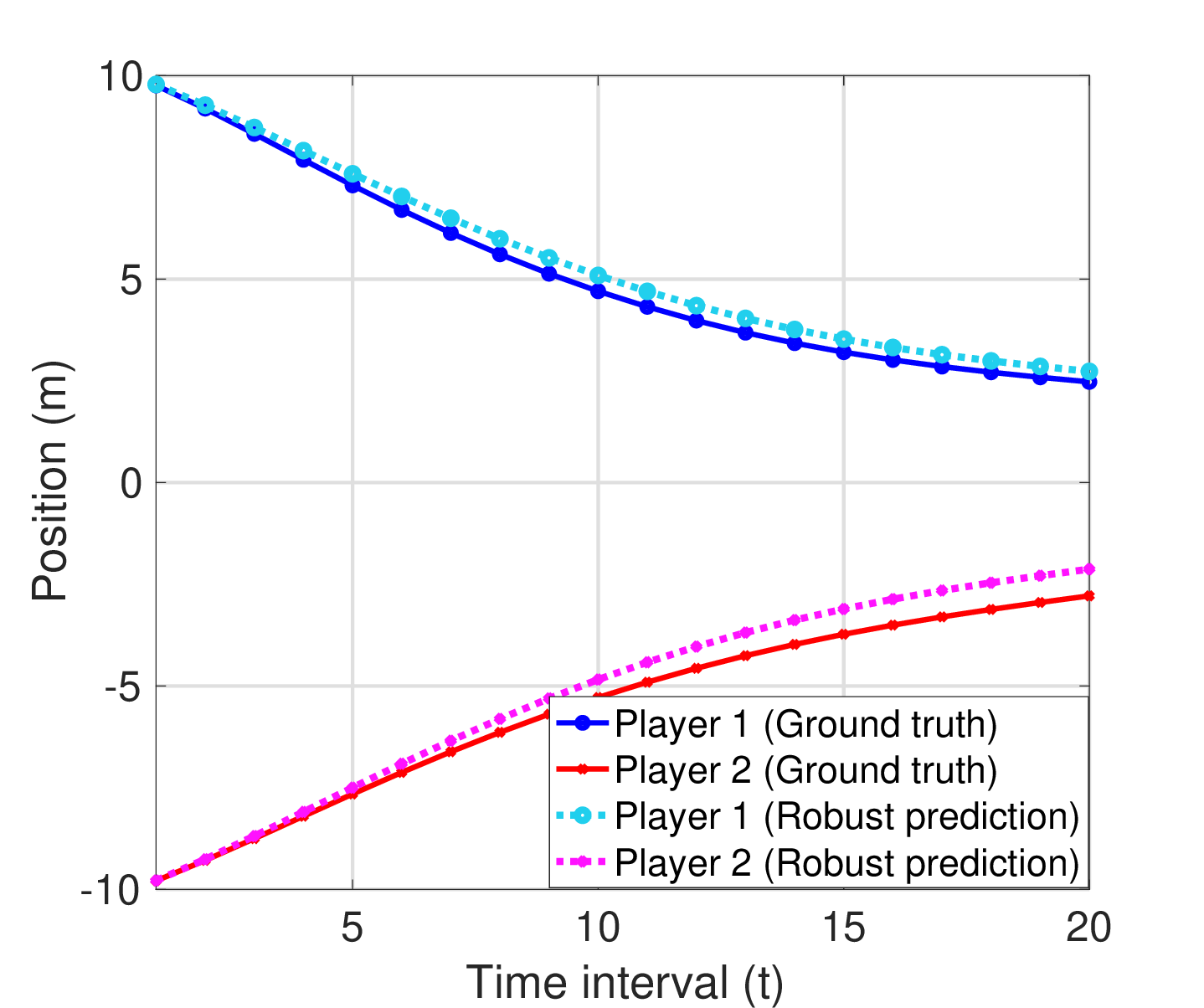} 
\caption{} \label{pos_pred} \end{subfigure} \hfill 
\begin{subfigure}[t]{0.32\linewidth} \centering \includegraphics[width=1.1\linewidth]{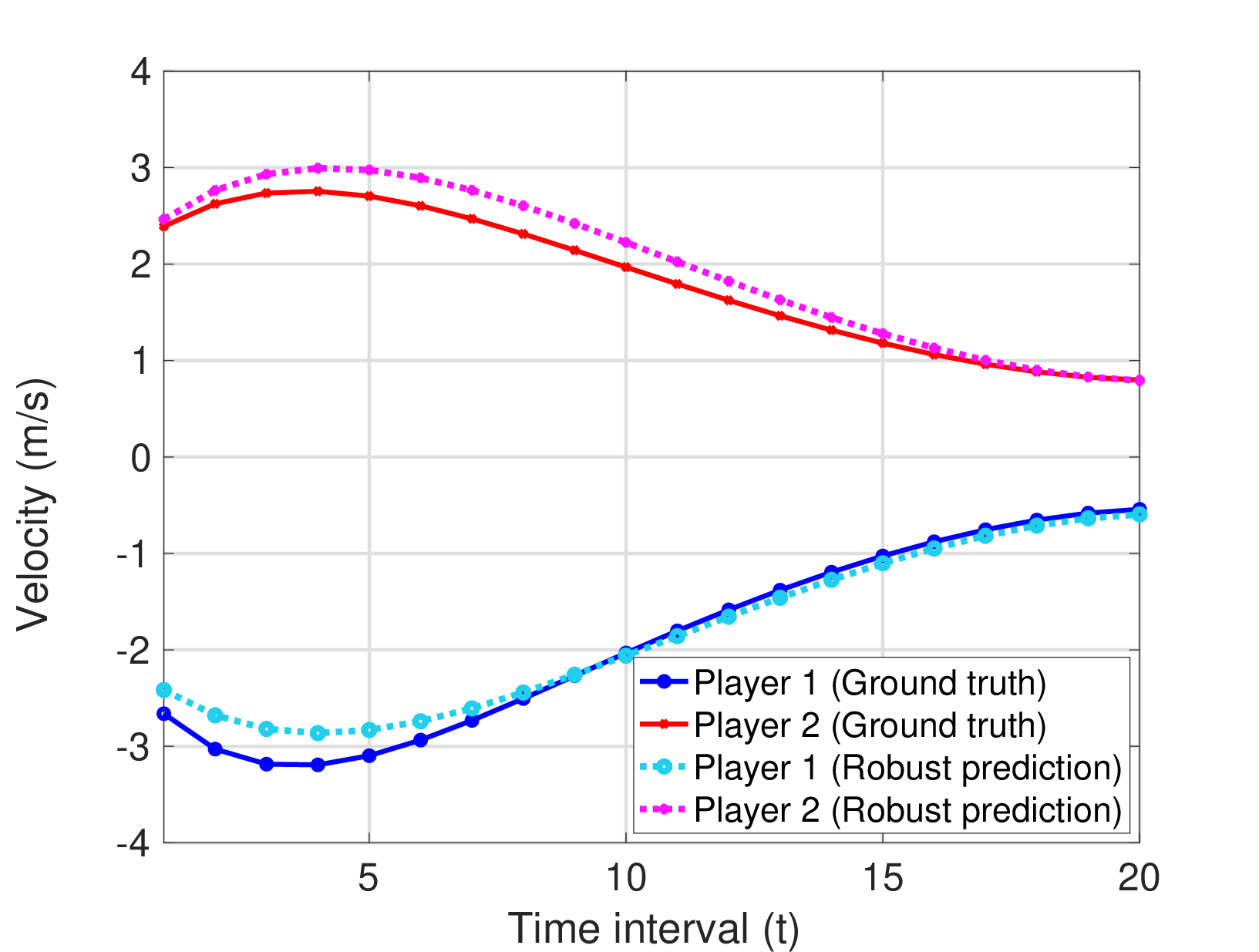} 
\caption{} \label{vel_pred} 
\end{subfigure} \hfill 
\begin{subfigure}[t]{0.32\linewidth} 
\centering 
\includegraphics[width=1.1\linewidth]{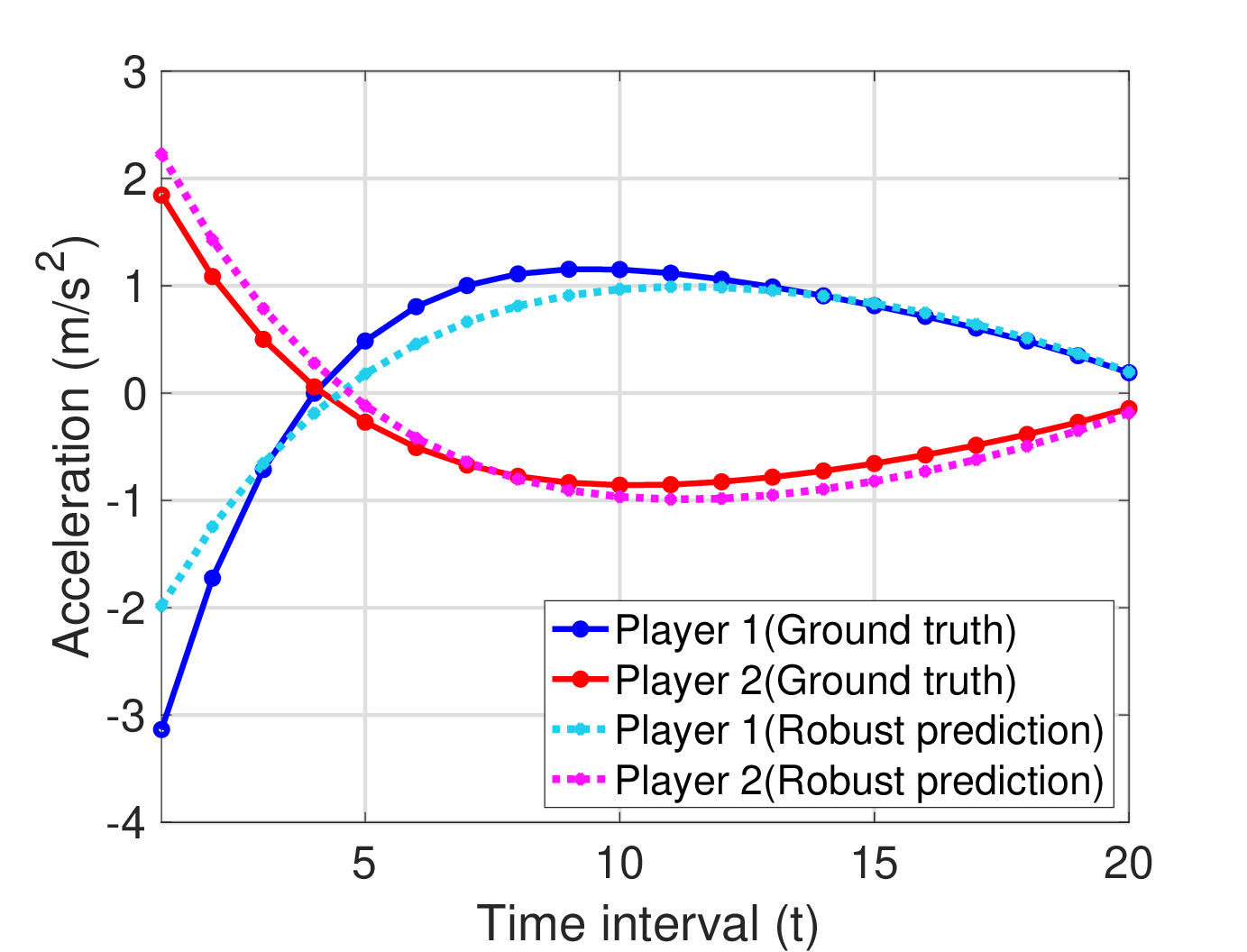} \caption{} 
\label{nash_pred} 
\end{subfigure} 
\caption{\small Closed-loop trajectory comparison of the ground truth versus the robust prediction for an unseen initial state: (a) Absolute position; (b) Velocity profiles; (c) Acceleration (Nash equilibrium control) profiles. The learned operator accurately captures the coupled braking negotiation.} 
\label{trajectory_fig} 
\end{figure}

\textit{Discussion of results.}
We train the model on a single trajectory ($T=20$) corrupted by Gaussian noise ($\sigma^2=0.05$) and solve \eqref{eq:regularized-full-problem}. Fig.~\ref{trajectory_fig} displays the predicted versus ground-truth trajectories for an unseen test episode. The robust prediction tightly tracks the ground-truth Nash equilibrium, accurately reproducing the yielding maneuver to maintain distance $D_s$ (Fig.~\ref{pos_pred}) and the asymmetric braking profiles induced by $R^1_{\text{true}} \neq R^2_{\text{true}}$ (Figs.~\ref{vel_pred}--\ref{nash_pred}). This confirms that the learned operator preserves the structural integrity of the inverse VI. To evaluate the proposed MAP-Tikhonov regularizer, we compare its Mean Absolute Error (MAE) against an unregularized inverse VI estimator over 30 test episodes. The MAE over the receding horizon $T$ is computed as $\text{MAE} = \frac{1}{ T |\mathcal{I}|} \sum_{i \in \mathcal{I}} \sum_{k=0}^{T-1} \left| \hat{u}^i_k(\boldsymbol{x}^{(j)}) - u^{i,*}_k(\boldsymbol{x}^{(j)}) \right|.$ Fig.~\ref{box_plot}(a) shows that as measurement noise ($\sigma^2$) increases, the unregularized estimator degrades significantly, exhibiting high median error and variance. Conversely, our robust estimator leverages the optimal penalty ($\rho_t$) to reject noise, maintaining a lower median MAE and tight variance bounds. Furthermore, Fig.~\ref{box_plot}(b) demonstrates the effect of horizon length ($T$). While a larger $T$ provides more data, it injects cumulative noise into the unregularized predictor. However, our MAP-Tikhonov formulation efficiently distills the extended spatio-temporal data, yielding monotonically decreasing MAE and tighter confidence intervals. This validates the efficacy of dynamic penalty allocation in practical online learning.

\begin{figure}[t]
    \centering
    \begin{subfigure}{0.48\linewidth}
        \centering
        \includegraphics[width=0.9\linewidth]{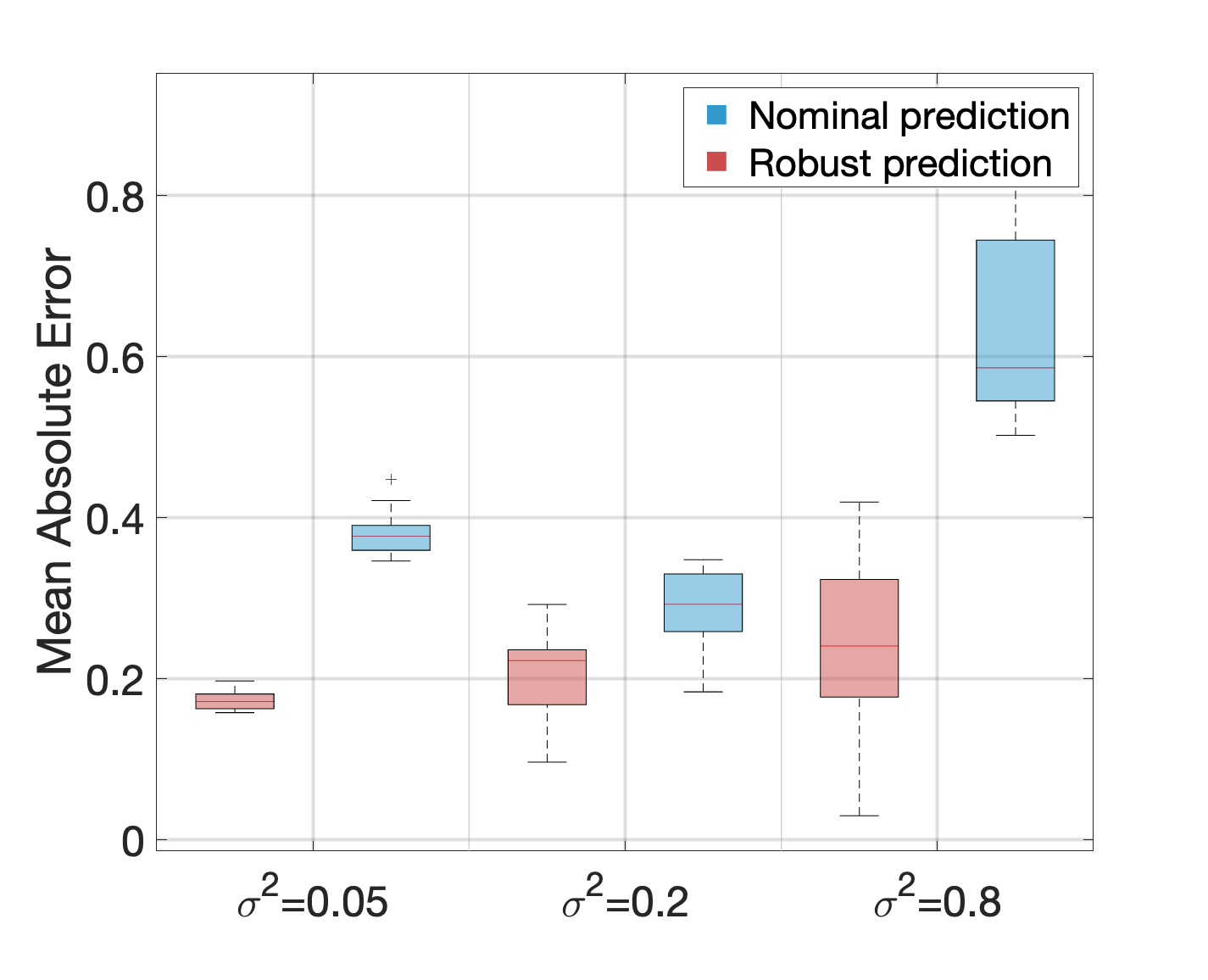}
        \caption{}
        \label{mae_var}
    \end{subfigure}
    \begin{subfigure}{0.48\linewidth}
        \centering
        \includegraphics[width=0.9\linewidth]{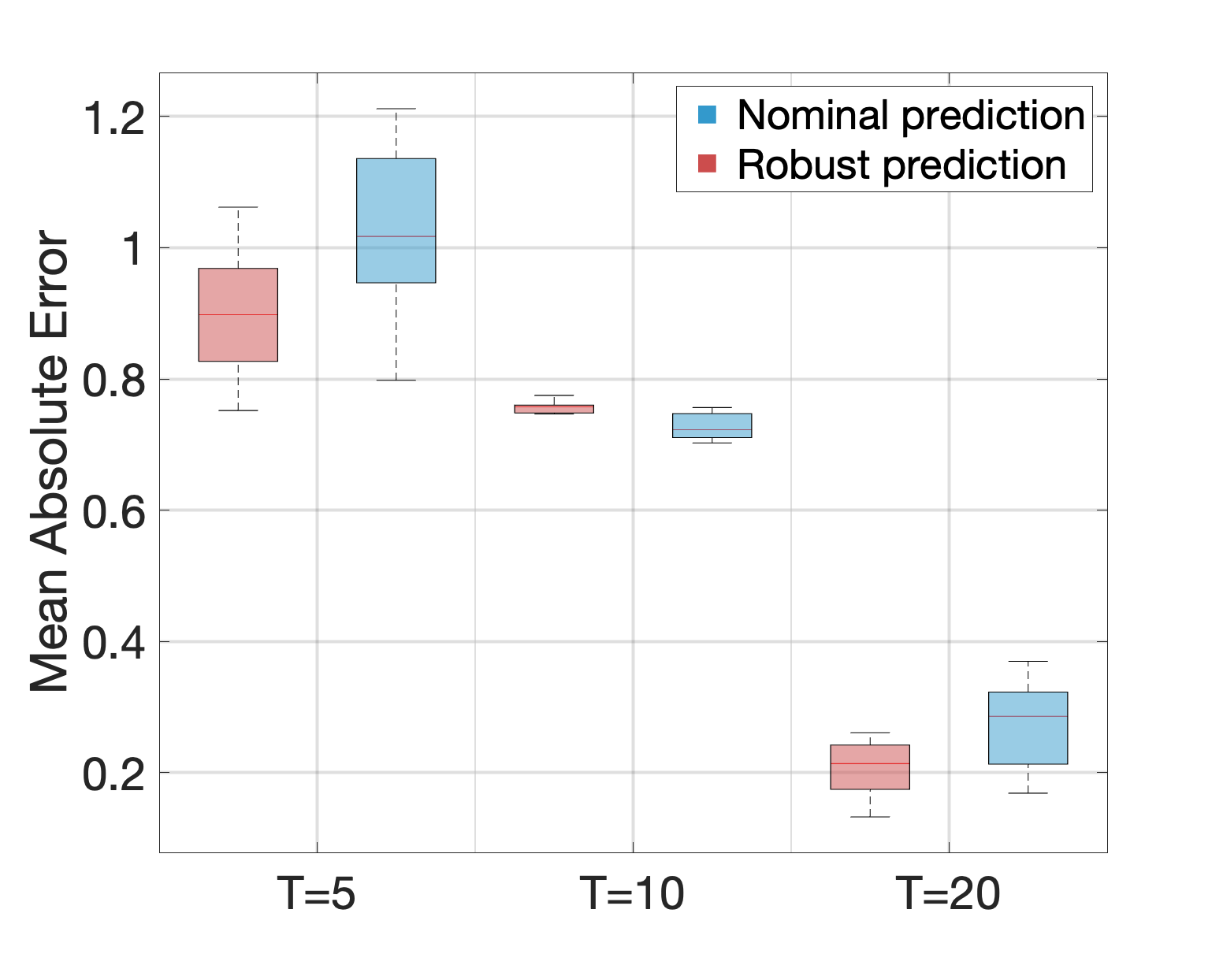} 
        \caption{}
        \label{hori_error}
    \end{subfigure}
    \caption{\small Prediction MAE over 30 test episodes: (a) MAE vs. measurement noise variance ($\sigma^2$) at fixed $T=20$; (b) MAE vs. receding horizon length ($T$) at fixed $\sigma^2 = 0.2$. The robust estimator consistently suppresses error variance, outperforming the baseline.}
    \label{box_plot}
\end{figure}

\section{Conclusion}
\label{sec:conclusion}
This article developed tractable convex formulations for inverse non-cooperative games based on Nikaido--Isoda and first-order losses. Strong monotonicity enables accurate equilibrium prediction for quadratic, non-quadratic, and linear-quadratic dynamic games with finite-sample guarantees, while a MAP-based receding-horizon estimator provides computationally efficient learning under measurement noise. Future work includes extending the offline guarantees beyond i.i.d. data to dynamically correlated trajectories, jointly inferring objectives and unknown feasibility sets, and robustifying the adaptive regularization against non-Gaussian, heavy-tailed, and adversarial disturbances.

\section{Appendix}
\subsection{Proof of Lemma \ref{loss_dominance}} \label{loss_lemma}
Given that $F_{\theta}$ is monotone, $J^i(u^i,u^{-i}\mid s)$ is convex in $u^i$. By definition, monotonicity implies that the symmetric Jacobian is positive semidefinite: $(\mathbf{J}_{F_{\theta}}+ \mathbf{J}_{F_{\theta}}^{\top})\succeq 0 $, from which the condition $\nabla^2_{u^i u^i}J^i(u^i,u^{-i}\mid s)\succeq 0$ directly follows. Thus, we state the first-order convexity inequality $J^i(u^i,u^{-i} \mid s) -J^i(y^i,u^{-i}\mid s)\leq \langle\nabla_{u^i}J^i(u^i,u^{-i}\mid s),(u^i-y^i)\rangle, \forall\: s \in \mathbb{S},\;y^i,u^i \in \mathcal{U}^i(s).$ Taking the supremum over $y^i$ on both sides yields $ J^i(u^i,u^{-i} \mid s) - \inf_{y^i \in \mathcal{U}^i(s)} J^i(y^i,u^{-i} \mid s) \leq\sup_{y^i \in \mathcal{U}^i(s)} \langle\nabla_{u^i}J^i(u^i,u^{-i}\mid s),(u^i-y^i)\rangle.$ Summing over all agents $i \in \mathcal{I}$, we recover $L^{\mathrm{NI}}_{\theta}(s,u) \leq L^{\mathrm{FO}}_{\theta}(s,u) \;\forall s \in \mathbb{S}, \; u \in \mathcal{U}(s),$ which concludes the proof. \hfill $\blacksquare$

\subsection{Proof of Proposition \ref{prop_qg} }\label{proof_prop_qg}
For brevity, we analyze the objective \eqref{ig_main} at a single time step $T=1$ (dropping the subscript $t$), focusing on the generic $i$-th agent. The inner minimization block of the NI loss evaluates as
\begin{align} \label{iner_min}
    \min_{y^i\in \mathcal{U}^i(\tilde{s}^i)} \quad &\frac{1}{2}(y^i)^{\top}R^i_{\theta}y^i + (\frac{1}{2}R^i_{\theta} \hat{u}^{-i} + r^i_{\theta}(\tilde{s}))^{\top} y^i \nonumber \\ 
    = \max_{\lambda^i\geq 0} \:\min_{y^i} \quad & \big[ \frac{1}{2}(y^i)^{\top}R^i_{\theta}y^i + (\frac{1}{2}R^i_{\theta} \hat{u}^{-i} + r^i_{\theta}(\tilde{s}))^{\top} y^i \big] + 
    (\lambda^i)^{\top} (G_i(\tilde{s})y^i-h_i(\tilde{s})).
\end{align}
Solving the inner unconstrained quadratic program over $y^i$ yields the algebraic root $y^i= -\frac{1}{2}(\hat{u}^{-i}+2(R^i_{\theta})^{-1}r^i_{\theta}(\tilde{s}) +2(R^i_{\theta})^{-1}G_i(\tilde{s})^{\top}\lambda^i )$. Substituting $y^i$ back into \eqref{iner_min} consolidates to $ -\frac{1}{2}\Big((\hat{u}^{-i})^{\top}R^i_{\theta}\hat{u}^{-i} +  (\hat{u}^{-i})^{\top}r^i_{\theta}(\tilde{s})\Big) -\min_{\lambda^i \geq 0} \frac{1}{2} \big[ (r^i_{\theta}(\tilde{s})+G_i(\tilde{s})^{\top}\lambda^i)(R^i_{\theta})^{-1}(r^i_{\theta}(\tilde{s})+G_i(\tilde{s})^{\top}\lambda^i) + (\hat{u}^{-i})^{\top}G_i(\tilde{s})^{\top}\lambda^i + 2(\lambda^i)^{\top}h_i(\tilde{s}) \big].$
We decouple this using an epigraph auxiliary variable $\gamma^i\in \mathbb{R}$, yielding the bounding constraint
\begin{align}\label{last_eq}
     (r^i_{\theta}(\tilde{s})+G_i(\tilde{s})^{\top}\lambda^i)(R^i_{\theta})^{-1}(r^i_{\theta}(\tilde{s})+G_i(\tilde{s})^{\top}\lambda^i) +  
     (\lambda^i)^{\top}(2h_i(\tilde{s})+G_i(\tilde{s})\hat{u}^{-i})-2\gamma^i \leq 0.
\end{align}
Because $R^i_{\theta}\succ 0$, \eqref{last_eq} is Schur-complement equivalent to the rigorous LMI block shown in \eqref{main_qg_prob}. Extending this globally across $[T]$ finalizes the proof. \hfill $\blacksquare$

\subsection{Proof of Theorem \ref{kernel_formulation}}\label{proof_prop_kernel}
Following the argument of the representation theorem in \cite[Theorem $8.7$]{paulsen2016introduction}, we can represent $\nabla_u^2g_\theta$ in \eqref{erm_kernel} with a sum of finite kernels as $\nabla_u^2g_\theta(u)= \sum_{t=1}^T\alpha_\theta^{t} \nabla^2_u k(u_t,u)$. This is because $\nabla^2_uk(u,\cdot)$ is a valid positive definite kernel. Furthermore, the RKHS-norm on $\nabla_u^2g_\theta$ is given by $\lVert \nabla_u^2g_\theta \rVert_{\mathcal{H}_k}^2= \alpha_\theta^{\top}K\alpha_\theta$, where $K$ is the Gram matrix having the elementwise component as $[K]_{ij}= \nabla_x\nabla_y^\top k(x,y)\mid_{(\tilde{u}_i,\tilde{u}_j)}= 2\gamma_k\Big(I-2\gamma_k(\tilde{u}_i-\tilde{u}_j)(\tilde{u}_i-\tilde{u}_j)^\top\Big)k(\tilde{u}_i,\tilde{u}_j)$, and $\alpha_\theta=\mathrm{col}((\alpha_\theta^j)_{j \in [T]})$. The FO loss constraint \eqref{fo_loss_defn} translates cleanly through strong convex duality. Bounding $\epsilon_t \geq L^{\mathrm{FO}}_{\theta}$, we deduce
\begin{align} \label{duality_eq}
        F_{\theta}(\tilde{s}_t,\tilde{u}_t)^{\top}\tilde{u}_t - \epsilon_t  &\leq \min_{u_t \in \mathcal{U}(\tilde{s})} \: F_{\theta}(\tilde{s}_t,\tilde{u}_t)^{\top}u_t\\ \nonumber
        &= \max_{\lambda_t \geq 0} \min_{u_t} \: \big[ F_{\theta}(\tilde{s}_t,\tilde{u}_t)^{\top}u_t + \lambda_t^{\top}(G(\tilde{s}_t)u_t-h(\tilde{s}_t)) \big] 
        = \max_{\substack{\lambda_t \geq 0 \\F_{\theta}(\tilde{s}_t,\tilde{u}_t)\\+G^{\top}(\tilde{s}_t)\lambda_t=0 }} \: -\lambda_t^{\top}h(\tilde{s}_t).
\end{align}
The remaining obstacle resides in the strong convexity constraint on the infinite-dimensional scalar field: $\nabla^2_u \phi_{\theta}(u)\succeq \lambda_{g}I_n$, mapping to $(\theta_R-\lambda_g)I_n \succeq -\nabla^2_{u} g_{\theta}(u)$. To bound the spectral norm strictly, we utilize the structural properties of RKHS, where derivative operators function as bounded linear maps $\sup_{u} \;\| \nabla^2_{u} g_{\theta}(u)\|_{\mathrm{op}} \leq C_k \|g_{\theta}\|_{\mathcal{H}},$ hence $(\theta_R-\lambda_g) \geq C_k \|g_{\theta}\|_{\mathcal{H}}$ is a sufficient condition. By direct substitution of the canonical RKHS norm $\|g_{\theta}\|_{\mathcal{H}}= \sqrt{\alpha_{\theta}^{\top}K_g\alpha_{\theta}}$, we deduce $(\theta_R-\lambda_g)\geq C_k\sqrt{\alpha_{\theta}^{\top}K_g\alpha_{\theta}}$, establishing an exact convex proxy. \hfill $\blacksquare$

\subsection{ Auxiliary Lemma} \label{lemma_ubh_prof}
\begin{lem}[Rademachar complexity] \label{ub_H}
    Define $D:=\sup_{u,v\in \mathcal{U}(\mathbf{s})}\|u-v\|$. Given the function class $\mathfrak{H}$ in \eqref{main_fc}, the empirical Rademacher complexity $\hat{R}_T(\mathfrak{H})$ is upper-bounded as 
    \begin{align}
        \hat{R}_T(\mathfrak{H}) \leq \frac{2(G+MD)(\lambda_R-\lambda_g)}{C_kT} \sqrt{\mathrm{Tr}(K)}.
    \end{align}
\end{lem}

   \textbf{Proof.} The first step is to relate Rademacher complexities of function classes $\mathfrak{H}$ and $\mathfrak{F}$. To do it, we need to know the existence of a Lipschitz constant of the FO loss $L^{\mathrm{FO}}_{\theta}$ (i.e., $\xi$ in the proof of Theorem \ref{nqg_theorem}). 

First, $F_\theta(x)-F_\theta(z)
=
A_\theta(x-z)
+
\sum_{j=1}^T
\alpha_\theta^j
\bigl(D_{jt}(x)-D_{jt}(z)\bigr).$
Since $u\mapsto ue^{-\gamma_k\|u\|^2}$ has bounded derivative, there exists \(C_\gamma>0\) such that $\|D_{jt}(x)-D_{jt}(z)\|
\le
4\gamma_k C_\gamma\|x-z\|.$ Hence $\|F_\theta(x)-F_\theta(z)\|
\le
\left(
\|A_\theta\|
+
4\gamma_k C_\gamma
\sum_{j=1}^T|\alpha_\theta^j|
\right)\|x-z\|,$ which gives $\|F_\theta(x)-F_\theta(z)\|
\le
M\|x-z\|.$ Next, $\|F_\theta(x)\|
\le
\|A_\theta\|\,\|x\|
+
\sum_{j=1}^T
|\alpha_\theta^j|
\|D_{jt}(x)\|.$ Since $\sup_{u\ge0}ue^{-\gamma_ku^2}
=
\frac{1}{\sqrt{2e\gamma_k}},$ we obtain $\|D_{jt}(x)\|
\le
\sqrt{\frac{2\gamma_k}{e}}.$ Therefore, using initial assumption of Theorem~\ref{nqg_theorem}, $\|F_\theta(x)\|
\le
\|A_\theta\|V
+
\sqrt{\frac{2\gamma_k}{e}}
\sum_{j=1}^T|\alpha_\theta^j|,$ and thus $\sup_{x\in \mathcal{U}(\mathbf{s})}\|F_\theta(x)\|
\le G$, where we denote upper-bound with $G$. Finally, for any \(x,z\in \mathcal{U}(\mathbf{s})\),
\[
\begin{aligned}
|L_{\theta}^{\mathrm{FO}}(x)
-
L_{\theta}^{\mathrm{FO}}(z)| & \le \sup_{y\in \mathcal{U}(\mathbf{s})} \left| \langle F_\theta(x),x-y\rangle - \langle F_\theta(z),z-y\rangle \right|
\\
&\le
\sup_{y\in \mathcal{U}(\mathbf{s})}
\left(
|\langle F_\theta(x),x-z\rangle|
+
|\langle F_\theta(x)-F_\theta(z),z-y\rangle|
\right)
\le
G\|x-z\|
+
MD\|x-z\|.
\end{aligned}
\]
Hence, $|L_{\theta}^{\mathrm{FO}}(x)
-
L_{\theta}^{\mathrm{FO}}(z)|
\le
(G+MD)\|x-z\|.$ Therefore, \(L_{\theta}^{\mathrm{FO}}\) is Lipschitz continuous on $\mathcal{U}(\mathbf{s})$, with Lipschitz constant $L_l= G+MD$. Now replacing $L^{\mathrm{FO}}_{\theta}$ with $\xi$, we invoke~\cite[Proposition $4.3$]{bach2024learning} to state the following result:
\begin{align} \label{bach_eq}
    \hat{R}_T(\mathfrak{H}) \leq L_l \hat{R}_T(\mathfrak{F})
\end{align}
The remaining part is to upper bound $\hat{R}_T(\mathfrak{F})$. For this, we use~\cite[Lemma $22$]{bartlett2002rademacher} to state $\hat{R}_T(\mathfrak{F}) \leq \frac{2(\lambda_R-\lambda_g)}{C_kT} \sqrt{\mathrm{Tr}(K)}$. Replacing it in \eqref{bach_eq} concludes the proof. \hfill $\blacksquare$ 
\subsection{Proof of Proposition~\ref{prop:MAP-equivalence}} \label{map_prop_proof}
By the differentiability of the optimizer mapping in Assumption~\ref{ass:local-MAP}(i), a first-order expansion around the noise-free observation gives $ \widehat\theta_t = \mathcal A_t(\bar y_t+\varrho_t)= \mathcal A_t(\bar y_t) + S_t\varrho_t + o_{\mathbb P}(\|\varrho_t\|). $ Since $\mathcal A_t(\bar y_t)=\theta_t^\star$ and $\varrho_t\sim\mathcal N(0,\sigma_m^2I)$, it follows that $ \widehat\theta_t \mid\theta_t^\star \;\dot{\sim}\; \mathcal N \left( \theta_t^\star, \sigma_m^2S_tS_t^\top \right), $ where $\dot{\sim}$ denotes equality under the local first-order approximation. Using the isotropic approximation in \eqref{eq:scalar-measurement-variance}, the likelihood is locally approximated by $ p(\widehat\theta_t\mid\theta_t^\star) \propto \exp\left( -\frac{1}{2\sigma_{\mathrm{meas},t}^2} \|\widehat\theta_t-\theta_t^\star\|_{\mathrm{Fr}}^2 \right). $ Moreover, the Gaussian transition model in Assumption~\ref{ass:local-MAP}(iv) gives $ p(\theta_t^\star\mid\theta_{t-1}) \propto \exp\left( -\frac{1}{2\sigma_{s,t}^2} \|\theta_t^\star-\theta_{t-1}\|_{\mathrm{Fr}}^2 \right). $ By conditional independence and Bayes' rule, $  p( \theta_t^\star \mid \widehat\theta_t,\theta_{t-1} ) \propto p(\widehat\theta_t\mid\theta_t^\star) p(\theta_t^\star\mid\theta_{t-1}).$ Therefore, the local MAP estimator in~\eqref{map_formula} satisfies 
\begin{equation} \label{eq:MAP-quadratic} 
\begin{aligned} \theta_t^{\mathrm{MAP}} \in \arg\min_{\theta\in\Theta_t} \bigg\{ & \frac{1}{2\sigma_{\mathrm{meas},t}^2} \|\theta-\widehat\theta_t\|_{\mathrm{Fr}}^2 
+ \frac{1}{2\sigma_{s,t}^2} \|\theta-\theta_{t-1}\|_{\mathrm{Fr}}^2 \bigg\}. 
\end{aligned} 
\end{equation}
By Assumption~\ref{ass:local-MAP}(iii), $ \frac{1}{2} \|\theta-\widehat\theta_t\|_{\mathrm{Fr}}^2 = \widetilde{\mathcal L}_t(\theta) - \widetilde{\mathcal L}_t(\widehat\theta_t) + o\!\left( \|\theta-\widehat\theta_t\|_{\mathrm{Fr}}^2 \right). $ The term $\widetilde{\mathcal L}_t(\widehat\theta_t)$ is independent of $\theta$ and can be discarded. Consequently, up to the stated second-order remainder, the MAP problem in~\eqref{eq:MAP-quadratic} becomes $ \arg\min_{\theta\in\Theta_t} \left\{ \widetilde{\mathcal L}_t(\theta) + \frac{ \sigma_{\mathrm{meas},t}^2 }{ 2\sigma_{s,t}^2 } \|\theta-\theta_{t-1}\|_{\mathrm{Fr}}^2 \right\}, $ which coincides locally with \eqref{eq:regularized-profiled-problem} under the update rule \eqref{eq:adaptive-rho}. \hfill $\blacksquare$

\bibliographystyle{apalike} 
\begin{small}
\bibliography{ref}
\end{small}

\newpage

\appendix

\end{document}